\documentclass[12pt]{article}
\pdfoutput=1
\usepackage{graphicx}
\usepackage{url}
\usepackage{rotating}
\usepackage{mathrsfs}
\usepackage{amssymb}
\usepackage{amsmath}
\usepackage{graphicx}
\usepackage{amsmath,amsfonts}
\usepackage{algorithm}
\usepackage{algorithmicx}
\usepackage{algpseudocode}
\usepackage[round]{natbib}
\usepackage{amscd}
\usepackage[mathscr]{eucal}
\usepackage{lineno}
\usepackage{extarrows}

\usepackage{listings}
\usepackage{tikz}

\newtheorem{thm}{Theorem}[section]
\numberwithin{thm}{section}
\newtheorem{remark}[thm]{Remark}
\newtheorem{definition}[thm]{Definition}
\newtheorem{cor}[thm]{Corollary}
\newtheorem{lemma}[thm]{Lemma}

\newenvironment{proof}{\noindent\\ \noindent\relax{\sc
     Proof}}{{\samepage\par\nopagebreak\hbox
     to\hsize{\hfill$\Box$}}}
\newcommand{\be}{\begin{equation}} \newcommand{\ee}{\end{equation}}
\newcommand{\bd}{\begin{displaymath}} \newcommand{\ed}{\end{displaymath}}
\newcommand{\ba}{\begin{align}} \newcommand{\ea}{\end{align}}
\newcommand{\baa}{\begin{align*}} \newcommand{\eaa}{\end{align*}}
\newcommand{\ben}{\begin{enumerate}} \newcommand{\een}{\end{enumerate}}
\newcommand{\bi}{\begin{itemize}} \newcommand{\ei}{\end{itemize}}

\newcommand{\E}[1]{\operatorname{E}\left[ #1 \right]}

\newcommand{\Expectation}[1]{\operatorname{E}\left[ #1 \right]}

\newcommand{\var}[1]{\operatorname{Var}\left[ #1 \right]}
\newcommand{\variance}[1]{\operatorname{Var}\left[ #1 \right]}
\newcommand{\cov}[2]{\operatorname{Cov}\left[ #1,#2 \right]}

\newcommand{\EcYn}[1]{\operatorname{E}\left[ #1 \vert \mathcal{Y}_{n}\right]}
\newcommand{\EY}[1]{\operatorname{E}_{Yule}\left[ #1 \right]}
\newcommand{\VarY}[1]{\operatorname{Var}_{Yule}\left[ #1 \right]}

\usepackage[normalem]{ulem}

\algnewcommand\And{\textbf{and}}

\begin{document}


\title{Exact and approximate limit behaviour of the Yule tree's cophenetic index}
\author{Krzysztof Bartoszek} 

\maketitle

\begin{abstract}
In this work we study the limit distribution of an appropriately normalized cophenetic index 
of the pure--birth tree conditioned on $n$ contemporary tips.
We show that this normalized phylogenetic balance index is a submartingale that converges
almost surely and in $L^{2}$. 
We link our work with studies on trees without branch lengths 
and show that in this case the limit distribution is 
a contraction--type distribution, similar to the Quicksort limit distribution. 
In the continuous branch case we suggest approximations to the limit distribution.
We propose heuristic methods of simulating from these distributions and it may be observed that these
algorithms result in reasonable tails. 
Therefore, we propose a way based on the quantiles of the derived distributions for hypothesis testing, whether
an observed phylogenetic tree is consistent with the pure--birth process. Simulating
a sample by the proposed heuristics is rapid, while exact simulation (simulating the
tree and then calculating the index) is a time--consuming procedure. We 
conduct a power study to investigate how well 
the cophenetic indices detect deviations from the Yule tree and apply the methodology to empirical phylogenies.
\end{abstract}


\section{Introduction}\label{secIntro}

Phylogenetic trees are now a standard when analyzing groups of species. They are inferred 
from molecular sequences by algorithms that often assume a Markov chain for mutations
at the individual positions of the genetic sequence \citep[e.g.][]{WEweGGra2005,JFel2004,ZYan2006}. 
Given a phylogenetic tree it is often
of interest to quantify the rate(s) of speciation and extinction for the studied species.
To do this one commonly assumes a birth--death process with constant rates. However,
the development of 
formal statistical tests whether a given tree
comes from a given branching process model is an open area of research
\citep[see the still relevant ``Work remaining'' part at the end of Ch. $33$ in][]{JFel2004}.
The reason for 
the apparent lack of widespread use of such tests \citep[but see][]{MBluOFra2005} 
could be the lack of a commonly
agreed on test statistic. This is as a tree is a complex object and there 
are multiple ways in which to summarize it in a single number. 

One proposed way of summarizing a tree is through indices that quantify
how balanced it is, i.e. how close is it to a fully symmetric tree. Two 
such indices have been with us for many years now: 
Sackin's \citep{MSac1972} and Colless' \citep{DCol1982}.
Alternatively, \citet{AMcKMSte2000} proposed to measure balance by counting cherries on the tree
and they showed that after appropriate centring and scaling, this index converges to the standard normal distribution
\citep[for examples of other indices see Ch. $33$ in][]{JFel2004}.

Recently, a new balance index was proposed---the cophenetic index \citep{AMirFRosLRot2013}.
The work here is inspired by private communication with evolutionary biologist Gabriel Yedid 
(current affiliation Nanjing Agricultural University, Nanjing, China) 
concerning the usage of the cophenetic index for significance testing 
of whether a given tree is consistent with the pure--birth process. He noticed
that simulated distributions of the index have much heavier tails than those of 
the normal and t distributions and hence, comparing 
centred and scaled cophenetic indices with the usual Gaussian or t quantiles is not appropriate for significance testing.
It would lead to a higher false positive rate---rejecting the null hypothesis of no extinction
when a tree was generated by a pure--birth process.

Our aim here is to propose an approach for working analytically with the cophenetic index,
especially to improve hypothesis tests for phylogenetic trees, i.e. how to recognize if the 
tree is out of the ``Yule zone'' \citep{GYanPAgaGYed2017}.
We show that there is a relationship between the cophenetic index and the Quicksort
algorithm. This suggests that the methods exploring \citep[e.g.][]{JFilSJan2000,JFilSJan2001,SJan2015} the limiting distribution of the 
Quicksort algorithm
can be an inspiration for studying analytical properties of the cophenetic index. 

The paper is organized as follows. In Section \ref{secCI} we formally define
the cophenetic index (for trees with and without branch lengths)
and present the most important results of the manuscript.
We define an associated submartingale
that converges almost surely and in $L^{2}$ (Thm. \ref{thmWnContConv}), propose
an elegant representation (Thm. \ref{thmWnZi}) and a very promising approximation (Def. \ref{defASlimPhi}).
Afterwards in Section \ref{secContrLimDist}, 
we show that in the discrete setting the limit law of the normalized
cophenetic index is a contraction--type distribution. Based on this we
propose alternative approximations to the limit law of the normalized (with branch lengths) 
cophenetic index. In Section \ref{secSignif}
we describe heuristic algorithms to simulate from these limit laws,
show simulated quantiles, explore the power of the cophenetic index to recognize
deviations from the Yule tree (comparing with Sackin's and Colless' indices' powers),
and apply the indices to example empirical data.
In Section \ref{secASbehaviour} we prove the claims presented in Section \ref{secCI} alongside 
other supporting results. Then, in Section 
\ref{sec2ndOrd} we study the second order properties of this decomposition 
and conjecture a Central Limit Theorem (CLT, Rem. \ref{remCLTconj}). 
We end the paper with Section \ref{secAltDesc}
by describing alternative representations of the cophenetic index.
\section{The cophenetic index and summary of main results}\label{secCI}
\citet{AMirFRosLRot2013} recently proposed a new balance index for phylogenetic trees.

\begin{definition}[\citet{AMirFRosLRot2013}]
For a given phylogenetic tree on $n$ tips and for each pair of tips
$(i,j)$ let $\tilde{\phi}_{ij}$ be the number of branches from the root
to the most recent common ancestor of tips $i$ and $j$. We then 
the define the \emph{discrete} cophenetic index as

$$
\tilde{\Phi}^{(n)} = \sum\limits_{1 \le i < j \le n}\tilde{\phi}^{(n)}_{ij}.
$$
\end{definition}
\citet{AMirFRosLRot2013} show that this index has a better resolution
than the ``traditional'' ones. In particular the cophenetic index
has a range of values of the order of $O(n^{3})$ while 
Colless' and Sackin's ranges have an order of $O(n^{2})$.
Furthermore, unlike the other two previously mentioned, 
$\tilde{\Phi}^{(n)}$ makes mathematical sense for trees that are not fully resolved
(i.e. not binary).
 
In this work we study phylogenetic trees with branch lengths and hence 
consider a variation of the cophenetic index.

\begin{definition}
For a given phylogenetic tree on $n$ tips and for each pair of tips
$(i,j)$ let $\phi_{ij}$ be the time 
from the most recent common ancestor of tips $i$ and $j$
to the root/origin (depending on the tree model) of the tree. We then 
define the \emph{continuous} cophenetic index as

$$
\Phi^{(n)} = \sum\limits_{1 \le i < j \le n}\phi^{(n)}_{ij}.
$$
\end{definition}
\begin{remark}
In the original setting, when the distance between two nodes was
measured by counting branches, \citet{AMirFRosLRot2013}
did not consider the edge leading to the root. 
In our work here, where our prime concern is with trees with 
random branch lengths, we include the branch leading to the root.
This is not a big difference, one just has to remember to add
to each distance between nodes the same exponential ($\exp(1)$---parametrization by the rate)
random variable (see Section \ref{secASbehaviour} for description of the tree's growth). 
\end{remark}

The results of the present manuscript are built around 
a scaled version of the cophenetic index which is an almost surely and $L^{2}$
convergent submartingale.
We first introduce some notation. 
Let $\mathcal{Y}_{n}$ be the $\sigma$--algebra containing all the information
on the Yule with $n$ tips tree and define

$$H_{n,m} := \sum_{k=1}^{n}1/k^{m}.$$
Below we present the main results concerning the cophenetic index, leaving 
the proofs and supporting theorems for Section \ref{secASbehaviour}.

\begin{thm}\label{thmWnContConv}
Consider a scaled cophenetic index

$$
W_{n}=\binom{n}{2}^{-1}\Phi^{(n)}.
$$
$W_{n}$ is a positive submartingale that converges almost surely and in $L^{2}$ to a 
finite first and second moment random variable.
\end{thm}

\begin{definition}
For $k=1,\ldots, n-1$ let us define $1_{k}^{(n)}$ as the indicator random variable taking the value of $1$ if a randomly sampled
pair of species coalesced at the $k$--th (counting from the origin of the tree) speciation event. 
\end{definition}
We know \citep[e.g.][]{KBarSSag2015,TreeSim1,MSteAMcK2001} that

\be
\mathbb{P}(1_{k}^{(n)}=1)=\Expectation{1_{k}^{(n)}} = 2\frac{n+1}{n-1}\frac{1}{(k+1)(k+2)} \equiv \pi_{n,k}.
\ee

\begin{definition}\label{dfnVi}
For $i=1,\ldots,n-1$ let us introduce the random variable

\be
V_{i}^{(n)} := \frac{1}{i}\sum\limits_{k=i}^{n-1}\EcYn{1_{k}^{(n)}}.
\ee
\end{definition} 

\begin{thm}\label{thmWnZi}
$W_{n}$ can be represented as

\be \label{eqWnZi}
W_{n} =\sum\limits_{i=1}^{n-1}V_{i}^{(n)}Z_{i},
\ee
where $Z_{1},\ldots,Z_{n-1}$ are i.i.d. $\exp(1)$ random variables.
\end{thm}

\begin{definition}\label{defASlimPhi}
Define the random variable $\overline{W}_{n}$ as

\be \label{eqASlimPhi}
\overline{W}_{n} =\sum\limits_{i=1}^{n-1}\E{V_{i}^{(n)}}Z_{i},
\ee
where $Z_{1},\ldots,Z_{n-1}$ are i.i.d. $\exp(1)$ random variables.
\end{definition}

\begin{remark}\label{remWbar}
Despite the apparent elegance, it is not straightforward to derive a Central Limit Theorem (CLT)
or limit statements concerning $W_{n}$ from the representation of Eq. \eqref{eqWnZi}. 
Initially one could hope 
\citep[based on ``typical'' results on limits for randomly weighted sums, e.g. Thm. $1$ of][]{ARosMSre1998}
that $W_{n}$ could 
converge a.s. to a random variable that has the same limiting distribution as 
$\overline{W}_{n}$.

Similarly, as in the proof of Thm. \ref{thmWnContConv} in Section \ref{secASbehaviour}, 
because $((n+2)(n-1)/(n(n+1))>1$, we have that $\overline{W}_{n}$
is an $L^{2}$ bounded submartingale

$$
\Expectation{\overline{W}_{n+1} \vert \overline{W}_{n}} = \frac{(n+2)(n-1)}{n(n+1)}\overline{W}_{n} + \frac{2}{n^{2}(n+1)} > \overline{W}_{n}.
$$
Hence, $\overline{W}_{n}$ converges almost surely.
Figure \ref{figSimulASlimPhi} can easily mislead one 
to believe in the equality of the limiting 
distributions of $W_{n}$ and $\overline{W}_{n}$. However, in Thm. \ref{thmVarWn} we can see 
that $\variance{W_{n}}$ and $\variance{\overline{W}_{n}}$ convergence to different limits. 
Therefore, $W_{n}$ and $\overline{W}_{n}$ cannot converge in distribution to the same limit. 
However, as we shall see in Section \ref{secSignif}, $\overline{W}_{n}$ provides a reasonable
approximation (and importantly extremely cheap, in terms of computational time and memory) 
to $W_{n}$ in the sense of their distributions. 
\end{remark}

\begin{figure}[!htp]
\centering
\includegraphics[width=0.6\textwidth]{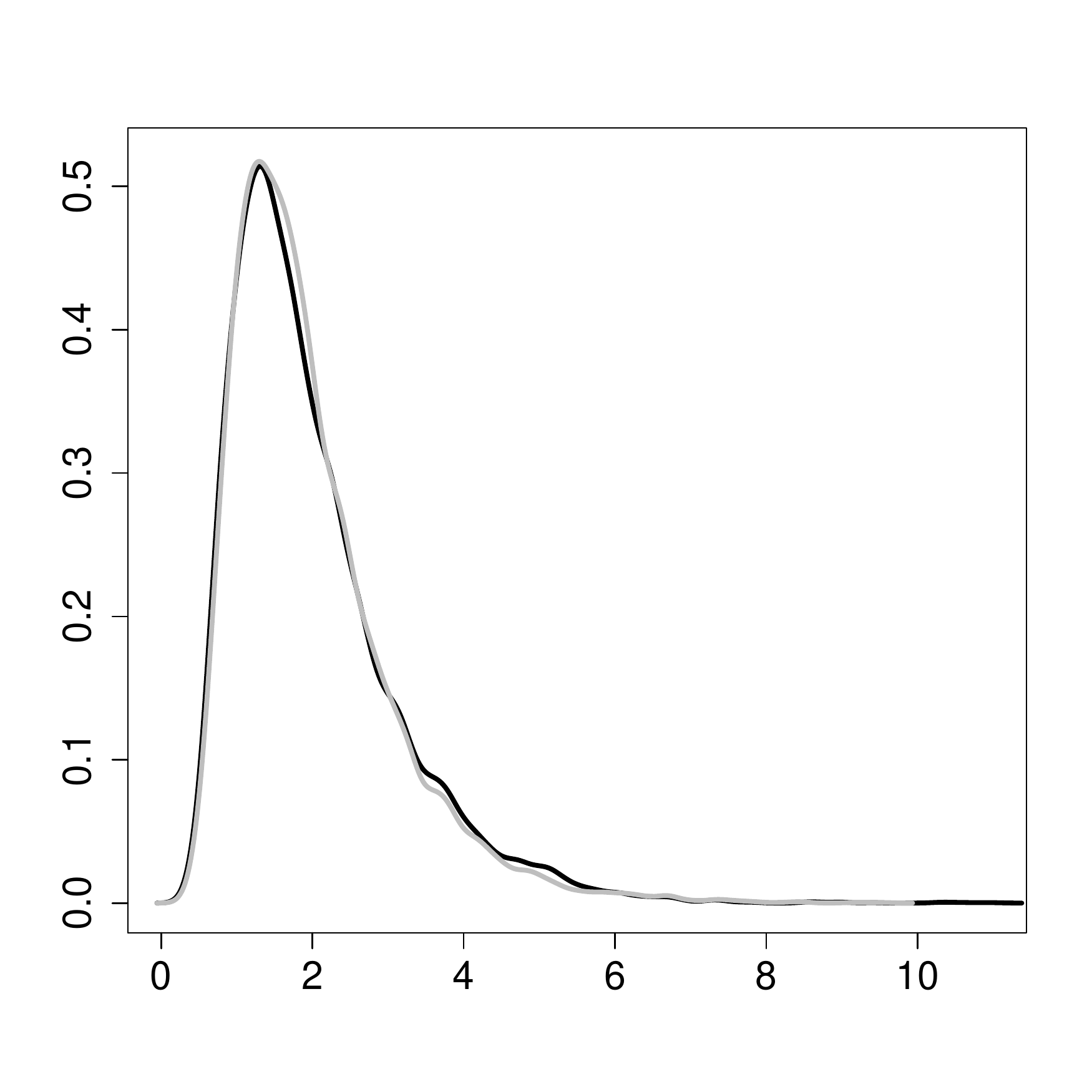}
\caption{
The curves are density estimates, via R's \citep{R} \texttt{density()} function, of $W_{n}$'s density (black)
and $\overline{W}_{n}$'s density (gray).
They are based on simulated values of $W_{n}$ from $10000$ simulated $500$--tips Yule trees with $\lambda=1$.
To obtain a sample from 
$\overline{W}_{n}$, independent $\exp(1)$ random variables were drawn.  
The simulated sample of $W_{n}$ has 
mean $2$, variance $1.214$, skewness $1.609$ and excess kurtosis $4.237$
while the 
simulated sample of $\overline{W}_{n}$ has 
mean $1.973$, variance $1.109$, skewness $1.634$ and excess kurtosis $4.159$.
It is obvious that $\E{W_{n}}=\E{\overline{W}_{n}}$, but we have shown
that their variances differ (simulations agree with Thm. \ref{thmVarWn}).
\label{figSimulASlimPhi}
}
\end{figure}

\begin{definition}
We naturally define the scaled discrete cophenetic index as

\be
\tilde{W}_{n}=\binom{n}{2}^{-1}\tilde{\Phi}^{(n)}.
\ee
\end{definition}
\begin{thm}\label{thmWnDiscConv}
$\tilde{W}_{n}$ is an almost surely and $L^{2}$ convergent submartingale.
\end{thm}

The applied reader will be most interested in how the results here can be practically used.
As written already in the Introduction balance indices are often used to provide
a single--number summary of the tree's shape. Such statistics can be then used e.g. to test 
if the tree is consistent with some null model (here the Yule model). Naturally, 
there has been extensive work on using different 
balance indices for significance testing \citep[e.g.][]{PAgaAPur2002,MBluOFra2005,GYanPAgaGYed2017}.
However previous works nearly always worked with indices that only considered the
topology and often obtained the rejection regions through direct simulations. 

Unfortunately, looking only at the tree's topology will not allow for distinguishing
between some models. In particular (as seen in Tab. \ref{tabCIpower}) there 
is no difference (from the topological indices perspective) between a Yule tree,
a constant rate birth--death tree and a coalescent tree. 
Hence, a temporal index that also takes into account the branch lengths should 
be used \citep[as indicated in the ``Work remaining'' section at the end of Ch. $33$ in][]{JFel2004}. A statistic based on $\Phi^{(n)}$ performs significantly better (but in
these cases still leaves a lot to be desired). However, $\Phi^{(n)}$
shows it true usefulness when employed to distinguish a biased speciation
\citep{MBluOFra2005} from a Yule model. \citet{MBluOFra2005} indicated that there is a
regime where topological indices fail completely. Table \ref{tabCIpower}
shows that in this setup (and also certain others) the temporal index in superior
in recognizing the deviation from the Yule tree.

Directly simulating a tree from a null model (Yule here) and then calculating
the index will of course give a sample from the correct null distribution.
However, this approach is costly both in terms of time and memory. Therefore,
if theoretical results that provide equivalent, asymptotic or approximate
representations of the index's law are available they could speed up 
any study by orders of magnitude. In fact this is clearly visible in
Tab. \ref{tabSimQuantResContDisc}, calculating the cophenetic index
directly from a sample of simulated pure--birth trees is over $170$ times
slower than considering $\overline{W}_{n}$. Even more dramatically
one can obtain a sample from an approximation to the equivalent representation of the asymptotic
distribution of $\tilde{\phi}^{(n)}$ (after normalization) nearly $3000$ times
faster than directly sampling the discrete cophenetic index. 

In Alg. \ref{algSignifTest} we describe how the presented here approach can be used for significance
testing. Then, in Section \ref{secSignif} we discuss in detail the required computational procedures,
present simulation results concerning the power of the tests and apply the tests to empirical data.
Preceding this computational Section is a characterization of the limit distribution of 
(normalized) $\tilde{\Phi}^{(n)}$ and another proposal to approximate the limit of (normalized)
$\Phi^{(n)}$. This section justifies the described simulation algorithms in Section \ref{secSignif}.

\begin{algorithm}[!htp]
\caption{Significance testing}\label{algSignifTest}
\begin{algorithmic}[1]
\State{\textbf{input}: A phylogenetic tree 
$\mathcal{T}^{(n)}$ with $n$ tips and significance level $\alpha$}
\State{\textbf{output}: A decision if the null hypothesis of 
$\mathcal{T}^{(n)}$ coming from a pure--birth process is rejected (\textbf{TRUE}) or not (\textbf{FALSE})}
\State Correct, when necessary, the tree for the speciation rate,
by multiplying all branch lengths by $\hat{\lambda}$, if cophenetic index with branch
lengths is used. \Comment{See Section \ref{sbsecPower}.}
\State Calculate $\Phi$, $\mathcal{T}^{(n)}$'s cophenetic index
\Comment{Exactly which version is used, 
$\Phi^{(n)}$, $\tilde{\Phi}^{(n)}$, $\Phi^{(n)}_{NRE}$, $\tilde{\Phi}^{(n)}_{NRE}$,
depends on the particular tree, if it has branch lengths or root edge}
\State Standardize $\Phi$ as $X=(\Phi-\EY{\Phi})/\sqrt{\VarY{\Phi}}$
\Comment{
$\EY{\Phi}$ and $\VarY{\Phi}$ depend which version
of the cophenetic index is considered. In Thm. \ref{thmEVarCI} all the possibilities are presented.}
\State Obtain the quantiles $q_{Yule}(\alpha/2)$, $q_{Yule}(1-\alpha/2)$ (if test is two--sided),
$q_{Yule}(\alpha)$ (left--tailed), $q_{Yule}(1-\alpha)$ (right--tailed)
of $X$ under the Yule model, i.e. $P(X \le q_{Yule}(\alpha)) = \alpha$.
\Comment{Exactly how to obtain the quantiles is a matter of which version of the cophenetic
index is used and computational resources (see Section \ref{secSignif}).
}
\If{$X$ is inside rejection region} 
    \Return \textbf{TRUE}
    \Else
    \State \Return \textbf{FALSE}
    \EndIf
\end{algorithmic}
\end{algorithm}

\begin{thm}\label{thmEVarCI}
A random variable with subscript NRE (no root--edge) indicates that this random variable comes from 
a tree lacking the edge leading to the root.

\be\label{eqEvarPhiCont}
\begin{array}{rcl}
\E{\Phi^{(n)}} & = & \binom{n}{2}\frac{2(n-H_{n,1})}{n-1}\sim n^{2}\\
\E{\Phi^{(n)}_{NRE}} & = & \binom{n}{2}\left(\frac{2(n-H_{n,1})}{n-1}-1\right)\sim \frac{1}{2}n^{2}\\
\var{\Phi^{(n)}} & = & 
\frac{\binom{n}{2}^{2}}{9n^{2}(n-1)^{2}}
\left(
12n^{2}(n^{2}-6n-4)H_{n-1,2} -9n^{4}+102n^{3}
\right. \\ &&  \left.
+51n^{2}-24nH_{n-1,1}-72n-72
\right)
\\ &\sim & \frac{1}{36}\left(2\pi^{2}-9\right)n^{4}\\
\var{\Phi^{(n)}_{NRE}} & = & 
\frac{\binom{n}{2}^{2}}{9n^{2}(n-1)^{2}}
\left(
12n^{2}(n^{2}-6n-4)H_{n-1,2} -9n^{4}+102n^{3}
\right. \\ &&  \left.
+51n^{2}-24nH_{n-1,1}-72n-72
\right)-\binom{n}{2}^{2}
\\ &\sim & \frac{1}{36}\left(2\pi^{2}-18\right)n^{4}

\end{array}
\ee

\be\label{eqEvarPhiDisc}
\begin{array}{rcl} 
\Expectation{\tilde{\Phi}^{(n)}} & = & \binom{n}{2}\left(\frac{4(n-H_{n,1})}{n-1}-1\right) \sim 3 n^{2}/2\\
\Expectation{\tilde{\Phi}^{(n)}_{NRE}} & = & \binom{n}{2}\left(\frac{4(n-H_{n,1})}{n-1}-2\right) \sim n^{2}\\ 
\var{\tilde{\Phi}^{(n)}} & = & \frac{1}{12}\left(n^{4}-10n^{3}+131n^{2}-2n \right)-4n^{2}H_{n,2}-6nH_{n,1} \sim n^{4}/12\\
\var{\tilde{\Phi}^{(n)}_{NRE}} & = & \frac{1}{12}\left(n^{4}-10n^{3}+131n^{2}-2n \right)-4n^{2}H_{n,2}-6nH_{n,1} \sim n^{4}/12
\end{array}
\ee

\end{thm}

\begin{proof}
The proof of the expectation part is due to \citet{AMirFRosLRot2013,SSagKBar2012}.
The variance of $\tilde{\Phi}^{(n)}$ is due to \citet{GCarAMirFRos2013,AMirFRosLRot2013}.
The variance of $\Phi^{(n)}$ 
is a consequence of the lemmata and theorems presented in Section \ref{sec2ndOrd}.
When the root edge is not included, then we have to decrease the expectation by $\binom{n}{2}$.
This is due to each pair of tips ``having'' the
root edge included in the cophenetic distance between them. In the case of branch lengths, the expectation
of the root edge, distributed as $\exp(1)$, is one. Without a root edge for the same reason
the variance of $\Phi^{(n)}$ has to be decreased by $\binom{n}{2}^{2}$. In the discrete case the root
edge has a deterministic length of $1$ and hence no effect on the variance.
\end{proof}

\section{Contraction--type limit distribution}\label{secContrLimDist}

Even though the representation of Eq. \eqref{eqWnZi} 
is a very elegant one, it is not obvious how to derive asymptotic properties 
of the process from it (compare Section \ref{sec2ndOrd}).
We turn to considering the recursive representation proposed by \citet{AMirFRosLRot2013} 

\be\label{eqPhiDiscRec}
\tilde{\Phi}^{(n)}_{NRE} = \tilde{\Phi}^{(L_{n})}_{NRE}+\tilde{\Phi}^{(R_{n})}_{NRE} + \binom{L_{n}}{2}+\binom{R_{n}}{2},
\ee
where $L_{n}$ and $R_{n}$ are the number of left and right daughter
tip descendants. 
Obviously $L_{n}+R_{n}=n$. 

From Eq. \eqref{eqPhiDiscRec} we will be able to deduce
the form of the limit of the process.  In the case with branch lengths
we attempt to approximate the cophenetic index with the following contraction--type law 

\be\label{eqPhiContRec}
\Phi^{(n)}_{NRE} = \Phi^{(L_{n})}_{NRE}+\Phi^{(R_{n})}_{NRE} + \binom{L_{n}}{2}T_{0.5}+\binom{R_{n}}{2}T'_{0.5},
\ee
where $T_{0.5}$, $T'_{0.5}$ are independent $\exp(2)$ random variables 
(we index with the mean to avoid confusion with $T_{2}$, Section \ref{secASbehaviour}, the time between the 
second and third speciation event which is also $\exp(2)$ distributed). 
These are the  branch lengths leading from the speciation point. The rationale
behind the choice of distribution is that
a randomly chosen internal branch of a conditioned Yule tree with rate $1$ 
is $\exp(2)$ distributed \citep[Cor. $3.2$ and Thm. $3.3$][]{TStaMSte2012}. 
This is of course an approximation, as we cannot expect that the laws
of the branch lengths with the depth of the recursion should
become indistinguishable from the law of the average branch.
In fact,
we should expect that the law of Eq. \eqref{eqPhiContRec} has to depend on $n$, i.e. the level of the recursion.
For larger $n$ the branches have distributions concentrated on smaller 
values, e.g. compare the randomly sampled root adjacent branch length law
\citep[Thm. $5.1$][]{TStaMSte2012} with the law of the average branch length. 

However, as we shall see simulations indicate that approximating
with the average law still could still yield acceptable heuristics, but not as good
as the approximation by $\overline{W}_{n}$.
We use the notation  $\tilde{\Phi}^{(n)}_{NRE}$, $\Phi^{(n)}_{NRE}$ to differentiate 
from $\tilde{\Phi}^{(n)}$, $\Phi^{(n)}$ where 
the root branch is included, i.e.

$$
\tilde{\Phi}^{(n)}=\tilde{\Phi}^{(n)}_{NRE}+\binom{n}{2}
~~\mathrm{and}~~
\Phi^{(n)}=\Phi^{(n)}_{NRE}+\binom{n}{2}T_{1},~~\mathrm{where}~T_{1}\sim exp(1).
$$
Define now

$$
Y^{(n)}=n^{-2}\left(\Phi^{(n)}_{NRE} - \Expectation{\Phi^{(n)}_{NRE}}\right)
~~~~
\tilde{Y}^{(n)}=n^{-2}\left(\tilde{\Phi}^{(n)}_{NRE} - \Expectation{\tilde{\Phi}^{(n)}_{NRE}}\right)
$$
and using Eqs. \eqref{eqEvarPhiCont} and \eqref{eqEvarPhiDisc} we obtain the following recursions

$$
\begin{array}{rcl}
Y^{(n)}&=&\left(\frac{L_{n}}{n} \right)^{2}Y^{(L_{n})} + \left(\frac{R_{n}}{n} \right)^{2}Y^{(R_{n})}
+n^{-2}\binom{L_{n}}{2}T_{0.5} + n^{-2}\binom{R_{n}}{2}T'_{0.5}
\\&&+n^{-2}\left(\Expectation{\Phi^{(L_{n})}_{NRE} \vert L_{n}} + \Expectation{\Phi^{(R_{n})}_{NRE} \vert R_{n}} - \Expectation{\Phi^{(n)}_{NRE}} \right)
\end{array}
$$
and

$$
\begin{array}{rcl}
\tilde{Y}^{(n)}&=&\left(\frac{L_{n}}{n} \right)^{2}\tilde{Y}^{(L_{n})} 
+ \left(\frac{R_{n}}{n} \right)^{2}\tilde{Y}^{(R_{n})}
+n^{-2}\binom{L_{n}}{2} + n^{-2}\binom{R_{n}}{2}
\\&&+n^{-2}\left(\Expectation{\tilde{\Phi}^{(L_{n})}_{NRE} \vert L_{n}} + \Expectation{\tilde{\Phi}^{(R_{n})}_{NRE} \vert R_{n}} - \Expectation{\tilde{\Phi}^{(n)}_{NRE}} \right).
\end{array}
$$
The process $\tilde{Y}^{(n)}$ is related to the
process $\tilde{W}_{n}$ as 

$$
\tilde{W}_{n}= 2(1+n^{-1})\tilde{Y}^{(n)}+\binom{n}{2}^{-1}\Expectation{\tilde{\Phi}^{(n)}_{NRE}}.
$$
In the continuous case we do not have an exact equality, we rather hope for

$$
W_{n} \approx 2(1+n^{-1})Y^{(n)}+\binom{n}{2}^{-1}\Expectation{\Phi^{(n)}_{NRE}} +T_{1}
$$
in some sense of approximation.
Hence, knowledge of the asymptotic behaviour of 
$Y^{(\infty)}$, $\tilde{Y}^{(\infty)}$ will immediately give us information about
$W^{(\infty)}$, $\tilde{W}^{(\infty)}$ in the obvious way

$$
\begin{array}{rcl}
\tilde{W}^{(\infty)} & = & 2\tilde{Y}^{(\infty)} +2 \\
W^{(\infty)} & \approx & 2Y^{(\infty)}+1+T_{1}.
\end{array}
$$

The processes $Y^{(n)}$, $\tilde{Y}^{(n)}$ 
look very similar to the scaled recursive representation
of the Quicksort algorithm \citep[e.g.][]{URos1991}.
In fact, it is of interest that, just as in the present work, a martingale proof first showed convergence 
of Quicksort \citep{MReg1989}, but then a recursive approach is required to show properties of the limit.
The random variable $L_{n}/n \to \tau \sim \mathrm{Unif}[0,1]$ weakly
and as weak convergence is preserved under continuous transformations
\citep[Thm. $18$, p. $316$][]{GGriDSti2009} we will have $(L_{n}/n)^{2} \to \tau^{2}$ weakly.
Therefore, we would expect 
the almost sure limits to satisfy the following equalities in distribution
(remembering the asymptotic behaviour of the expectations)

\be\label{eqLimY}
Y^{(\infty)} \stackrel{\mathcal{D}}{=} \tau^{2} Y^{'(\infty)} 
+ (1-\tau)^{2}Y^{''(\infty)}
+\frac{1}{2}\tau^{2}T_{0.5}+\frac{1}{2}(1-\tau)^{2}T'_{0.5} 
-\tau(1-\tau),
\ee
and 

\be\label{eqLimtY}
\tilde{Y}^{(\infty)} \stackrel{\mathcal{D}}{=} \tau^{2} \tilde{Y}^{'(\infty)} + (1-\tau)^{2}\tilde{Y}^{''(\infty)} +
\frac{1}{2} -3\tau(1-\tau) 
\ee
where $\tau$ is uniformly distributed on $[0,1]$, 
$Y^{(\infty)}$, $Y^{'(\infty)}$ and $Y^{''(\infty)}$ are identically distributed random variables,
so are $\tilde{Y}^{(\infty)}$, $\tilde{Y}^{'(\infty)}$ and $\tilde{Y}^{''(\infty)}$,
and $Y^{'(\infty)}$, $Y^{''(\infty)}$, $\tilde{Y}^{'(\infty)}$ and $\tilde{Y}^{''(\infty)}$ are independent.
Following \citet{URos1991}'s approach it turns out that the limiting distributions do satisfy the equalities
of Eqs. \eqref{eqLimY} and \eqref{eqLimtY}.

Let $D$ be the space of distributions with zero first moment and finite second moment. 
We consider on $D$ the Wasserstein metric

$$
d(F,G) = \inf\limits_{X\sim F, Y\sim G} \Vert X- Y \Vert_{L^{2}}.
$$

\begin{thm}\label{thm2.1}
Let $F\in D$ and assume that $Y, Y' \sim F$, $\tau \sim \mathrm{Unif}[0,1]$, $T_{0.5}, T_{0.5}' \sim \mathrm{exp}(2)$ and
$Y, Y', \tau, T, T'$ are all independent. Define transformations $S_{1}:D\to D$, $S_{2}:D\to D$ as

\be
S_{1}(F) = \tau^{2} Y + (1-\tau)^{2}Y' +\frac{1}{2}\tau^{2}T_{0.5}+\frac{1}{2}(1-\tau)^{2}T'_{0.5} -\tau(1-\tau),
\ee
and

\be
S_{2}(F)= \tau^{2} Y + (1-\tau)^{2}Y^{'} + \frac{1}{2} -3\tau(1-\tau) 
\ee
respectively. Both transformations $S_{1}$ and $S_{2}$ are contractions on $(D,d)$ and converge exponentially
fast in the $d$--metric to the unique fixed points of $S_{1}$ and $S_{2}$ respectively.
\end{thm}

\begin{remark}
The proof of Thm. \ref{thm2.1} is the same as \citet{URos1991}'s proof of his Thm. $2.1$. 
However, compared to the Quicksort algorithm \citep{URos1991} we will have a $\sqrt{2/5}$ upper bound
on the rate of decay instead of $\sqrt{2/3}$.
This speed--up should be expected as we have $\tau^{2}$ and $(1-\tau)^{2}$ instead
of $\tau$ and $(1-\tau)$. 
Thm. \ref{thm2.1} can also be seen as a consequence of \citet{URos1992}'s more general
Thms. $3$ and $4$. 
The rate of convergence is also a consequence of the general contraction lemma
\citep[Lemma 1,][]{URosLRus2001}.
\end{remark}

Now, using Lemmata \ref{lemCCDiscdiff}, \ref{lemCCdiff}
\citep[their proofs in \ref{appLemProofs} differ only in detail from the proof of Prop. $3.2$ in][]{URos1991}
and arguing in the same way as \citet{URos1991} did in his Section $3$, especially
his proof of his Thm. $3.1$ we obtain that
 $Y^{(n)}$ and $\tilde{Y}^{(n)}$
converge in the Wasserstein $d$--metric to
$Y^{(\infty)}$ and $\tilde{Y}^{(\infty)}$
whose laws are fixed points of $S_{1}$ and $S_{2}$ respectively.
A minor point should be made. 
Here, we will have $(i/n)^{4}$ instead of $(i/n)^{2}$ in a counterpart of \citet{URos1991}'s Prop $3.3$.

\begin{remark}
One may directly obtain from the recursive representation that
$\E{Y^{(\infty)}}=E{\tilde{Y}^{(\infty)}}=0$, $\var{Y^{(\infty)}}=1/16=0.0635$
and $\var{\tilde{Y}^{(\infty)}}=1/12$. We can therefore, see that in the discrete
case the variance agrees. However, in the continuous case we can see that it slightly
differs 

$$\var{(W_{n}-T_{1})/2}\to \pi^{2}/18-0.5\approx0.048.$$
\end{remark}

\begin{remark}\label{remTgamma}
One can of course calculate what the mean and variance of $T_{0.5}$, $T_{0.5}'$ should
be so that $\E{Y^{(\infty)}}=0$ and $\var{Y^{(\infty)}}=\var{(W_{n}-T_{1})/2}$. 
We should have $\E{T_{0.5}}=\E{T_{0.5}'}=0.5$ and 
$\var{T_{0.5}}=\var{T_{0.5}'}=\pi^{2}/3-25/8$. This, in particular, means 
that these branch lengths cannot be exponentially distributed. 
We therefore, also experimented by drawing $T_{0.5}$, $T_{0.5}'$
from a gamma distribution with rate equalling $1/(2(\pi^{2}/3-25/8))$ and
shape equalling $\pi^{2}/6-25/16$. However, this significantly increased the duration of the computations
but did not result in any visible improvements
in comparison to Tab. \ref{tabSimQuantResContDisc}.
\end{remark}

\section{Significance testing}\label{secSignif}

\subsection{Obtaining the quantiles}
Algorithm \ref{algSignifTest} requires knowledge of the quantiles of the underlying
distribution in order to define the rejection region.
Unfortunately, an analytical form of the density of any scaled cophenetic index
is not known so one will have to resort to some sort of simulations
to obtain the critical values. Directly simulating
a large number of pure--birth trees can take an overly long time, measured in minutes (on a modern machine
with a large amount of memory, or hours on an older one).
Fortunately, the cophenetic index can be calculated in $O(n)$ time \citep[Cor. 3][]{AMirFRosLRot2013}
and such a tree--traversing algorithm was employed to obtain $\Phi^{(n)}$ and $\tilde{\Phi}^{(n)}$.
On the other hand, the suggestive (but wrong) approximations of Eq. \eqref{eqASlimPhi} and contraction limiting
distributions Eqs. \eqref{eqLimY} and \eqref{eqLimtY} are significantly faster
to simulate, see Tab. \ref{tabSimQuantResContDisc}. 

Simulating from the approximate Eq. \eqref{eqASlimPhi} is straightforward. One simply draws 
$n-1$ independent $\exp(1)$ random variables. Simulating random variables
satisfying  Eqs. \eqref{eqLimY} and \eqref{eqLimtY} is more involved and 
it may be possible to develop an exact rejection algorithm
\citep[cf.][]{LDevJFilRNei2000}.
Here, we choose simple, approximate but still effective, heuristics in order to demonstrate the usefulness of the approach
for significance testing.

We now describe algorithms (Algs. \ref{algFsimPop} and \ref{algFsimRecurs})
for simulating from a more general distribution, $F$, that satisfies

\be\label{eqContrDist}
Y \stackrel{\mathcal{D}}{=} g_{1}(\tau) Y' + g_{2}(\tau)Y'' + C(\tau,\theta),
\ee
where $Y,Y',Y'' \sim F$, $Y',Y'',\tau,\theta$ are independent, $\tau \sim F_{\tau}$,
$\theta \sim F_{\theta}$ is some random vector, $g_{1},g_{2}: \mathbb{R} \to \mathbb{R}$
and $C:\mathbb{R}^{p} \to \mathbb{R}$ for some appropriate $p$ that depends on $\theta$'s dimension.
Of course in our case here 
we have $\tau\sim\mathrm{Unif}[0,1]$,
$g_{1}(\tau)=\tau^{2}$, $g_{2}(\tau)=(1-\tau)^{2}$,
 
$$C(\tau,T,T')=\tau^{2}T/2+(1-\tau)^{2}T'/2-\tau(1-\tau)$$
and 

$$C(\tau)=1/2-3\tau(1-\tau)$$
for $\Phi^{(n)}$, $\tilde{\Phi}^{(n)}$ respectively. 
Of course, $T$, $T'$ are independent and $\exp(2)$ distributed.
If one considers also the root edge, then to the simulated random variable one needs
to add $T_{1} \sim \exp(1)$ when simulating $n^{-2}\Phi^{(n)}$ or appropriately $1$ if one considers 
$n^{-2}\tilde{\Phi}^{(n)}$.

\begin{algorithm}[!htp]
\caption{Population approximation}\label{algFsimPop}
\begin{algorithmic}[1]
\State Initiate population size $N$
\State Set $P[0,1:N]=Y_{0}$ \Comment{Initial population}
\For{$
i=1$ to $i_{max}$}
    \State $f_{i-1}=$\texttt{density(}$P[i-1,]$\texttt{)} \Comment{density estimation by \texttt{R}}
    \For{$j=1$ to $N$}
    \State Draw $\tau$ from $F_{\tau}$
    \State Draw $\theta$ from $F_{\theta}$
    \State Draw $Y_{1}$, $Y_{2}$ independently from $f_{i-1}$
    \State $P[i,j]=g_{1}(\tau)Y_{1}+g_{2}(\tau)Y_{2}+C(\tau,\theta)$	
    \EndFor
\EndFor
\State \Return{$
P[i_{max},]$} \\ \Comment{Add root branch ($\exp(1)$ or $1$) if needed for each individual.}
\end{algorithmic}
\end{algorithm}

\begin{algorithm}[!htp]
\caption{Recursive approximation}\label{algFsimRecurs}
\begin{algorithmic}[1]
\Procedure{Yrecursion}{$n$, 
$Y_{0}$}
    \If{$n=0$}  
    \State $Y_{1}=Y_{0}$, $Y_{2}=Y_{0}$    
    \ElsIf {$n=1$ \textbf{and} $Y_{0}=0$}
    \State Draw $\tau_{1}$, $\tau_{2}$ independently from $F_{\tau}$
    \State Draw $\theta_{1}$, $\theta_{2}$ independently  from $F_{\theta}$
    \State $Y_{1}=C(\tau_{1},\theta_{1})$
    \State $Y_{2}=C(\tau_{2},\theta_{2})$
    \Else
    \State $Y_{1}=$\Call{Yrecursion}{$n-1$
    , $Y_{0}$}
    \State $Y_{2}=$\Call{Yrecursion}{$n-1$
    , $Y_{0}$}
    \EndIf	
    \State Draw $\tau$ from $F_{\tau}$
    \State Draw $\theta$ from $F_{\theta}$
    \State \Return{$
    g_{1}(\tau) Y_{1} + g_{2}(\tau)Y_{2} + C(\tau,\theta),$}
\EndProcedure
\State \Return{\Call{Yrecursion}{$
N$, $Y_{0}$}} \\ \Comment{Add root branch ($\exp(1)$ or $1$) if needed.}
\end{algorithmic}
\end{algorithm}
The recursion of Alg. \ref{algFsimRecurs} for a given realization of $\tau$ and $\theta$ random variables
can be directly solved. However, from numerical experiments implementing Alg. \ref{algFsimRecurs}
iteratively seemed computationally ineffective.

In Tab. \ref{tabSimQuantResContDisc} we report on the simulations from the different distribution.
For each distribution we draw a sample of size $10000$ and repeat this $100$ times.
We compare the quantiles from the different distributions.
We can see that the approximation of $\overline{W}_{n}$ for $W_{n}$ is a good one 
and can be used when one needs to work with the distribution of the cophenetic index
with branch lengths. In the case of the discrete cophenetic index we have found
an exact limit distribution which is a contraction--type distribution. Therefore,
one can relatively quickly simulate a sample from it without the need to 
do lengthy simulations of the whole tree and then calculations of the cophenetic
index. Unfortunately, this contraction approach does not seem to give 
such good results in the Yule tree with branch lengths case. We used
an approximation when constructing the contraction. Instead of taking
the law of the length of two daughter branches, we took the law
of an random internal branch. This induces a difference between
the tails of the distributions that is clearly visible in the simulations.
Even at the second moment level there is a difference. We calculated (Thm. \ref{thmVarWn})
that $\var{W_{n}}\to 2\pi^{2}/9-1 \approx 1.193$, $\var{\overline{W}_{n}} \to 4\pi^{2}/3-12\approx 1.159$
while $\var{2Y^{(n)}+T_{1}}=1.25$. Therefore, the approximation by $\overline{W}_{n}$
seems better already at the second moment level. Generally
if one cannot afford the time and memory to simulate a large sample of Yule tree, simulating
$\overline{W}_{n}$ values seems an attractive option, as the discrepancy between
the two distributions seems small.

In Fig. \ref{figSimulRecurs} we compare the density estimates of (scaled and centred)
both continuous and discrete branches cophenetic indices and their respective contraction--type limit 
distributions. The density estimates generally agree but we know from Tab. \ref{tabSimQuantResContDisc}
that for $\Phi^{(n)}$ this is only an approximation. We simulated $10000$ Yule trees
and hence we report only the quantiles between $2.5\%$ and $97.5\%$. 
Quantiles further out in the tails seemed less accurate and hence are not included in
the table. Similarly, we can see less correspondence between the different
estimates of kurtosis. This statistic relies on fourth
moments and hence is more sensitive to the tails. On the other hand we can see
much greater Monte Carlo error for the kurtosis in all simulations, including
the setup where the values are extracted directly from Yule trees.
The values for $\tilde{\Phi}^{(n)}$ seem more similar to values from
the Yule tree. We should expect this as here we have shown an exact limit distribution.

An overall assessment of the quantiles is given by the root--mean--square--error (RMSE) row in 
Tab. \ref{tabSimQuantResContDisc}. We consider the quantiles at the 
$\alpha_{1}=0.001$, $\alpha_{2}=0.005$, $\alpha_{3}=0.01$, $\alpha_{4}=0.025$, $\alpha_{5}=0.05$, 
$\alpha_{6}=0.95$, $\alpha_{7}=0.975$, $\alpha_{8}=0.99$, $\alpha_{9}=0.995$, $\alpha_{10}=0.999$ levels.
The RMSE is defined as 

\be\label{eqRMSE}
\mathrm{RMSE}=\left(\frac{1}{100}
\sum\limits_{j=1}^{100}
\left(
\sum\limits_{i=1}^{5} \left(\alpha_{i}-\alpha_{i-1} \right)\left(\hat{q}_{j}(\alpha_{i}) - q(\alpha_{i}) \right)^{2}
+
\sum\limits_{i=6}^{10} \left(\alpha_{i+1}-\alpha_{i} \right)\left(\hat{q}_{j}(\alpha_{i}) - q(\alpha_{i}) \right)^{2}
\right)
\cdot (0.1)^{-1}\right)^{\frac{1}{2}}
\ee
with dummy levels $\alpha_{0}=0$ and $\alpha_{11}=1$. The $(0.1)^{-1}$ normalizes the whole mean--square--error. We only look at the error 
at the tails, so we correct by the fraction of the distributions' support that we consider.
As a proxy for the true quantiles we take the pooled values 
(as explained in Tab. \ref{tabSimQuantResContDisc}) from the ``Yule columns''.
The $j$ index runs over the $100$ repeats of the simulations.

The RMSE, when using $\overline{W}_{N}$, seems to be on the level of the RMSE
of the ``direct simulations''. $Y^{(N)}$ has an error
of about twice the size (both simulation methods). Looking at $\tilde{Y}^{(N)}$
one can see that the RMSE is exactly on the level of the ``Yule column's'' RMSE. 
This is even though we used a recursion of level $10$, while an exact match of distributions should 
take place in the limit (infinite depth recursion). However, the rapid, exponential convergence 
of the contraction seems to make any differences invisible, already at this recursion level.

\begin{sidewaystable}[!htp]
\centering
{\small
\begin{tabular}{cccccccccc}
\hline
\multicolumn{6}{c}{$
\left(\sqrt{\var{\Phi^{(n)}}}\right)^{-1}(\Phi^{(n)}-\Expectation{\Phi^{(n)}})$ limit approximation} & \multicolumn{4}{c}{$
\left(\sqrt{\var{\tilde{\Phi}^{(n)}}}\right)^{-1}(\tilde{\Phi}^{(n)}-\Expectation{\tilde{\Phi}^{(n)}})$ limit approximation} \\ 
& Yule & $\mathcal{N}(0,1)$ & $\overline{W}_{N}$ & $Y^{(N)}$ Alg. \ref{algFsimPop} & $Y^{(N)}$ Alg. \ref{algFsimRecurs}  & Yule & $\mathcal{N}(0,1)$ & $\tilde{Y}^{(N)}$ Alg. \ref{algFsimPop} & $\tilde{Y}^{(N)}$ Alg. \ref{algFsimRecurs}\\
\noalign{\smallskip}\hline\noalign{\smallskip}
Run time & $690.918$s & --- & $3.905$s & $0.318$s & $110.021$s & $698.269$s & --- &$0.233$s & $44.358$s \\
Avg. $(=0)$ & $-0.023,0.029$ & $0$ & $-0.025,0.024$ & $-0.024,0.026$ & $-0.019,0.026$ & $-0.02,0.025$ & $0$  & $-0.033,0.032$ & $-0.028,0.02$ \\
 & $0.002$ & $0$ & $-0.001$ & $0.006$ & $0$ & $0$ & $0$  & $-0.001$ & $0$ \\
Var. $(=1)$ & $0.946,1.074$  & $1$ & $0.921,1.025$ & $0.928,1.072$ & $0.932,1.061$ & $0.939,1.038$ & $1$ & $0.953,1.087$ & $0.931,1.047$  \\
& $1.003$  & $1$ & $0.97$ & $1.014$ & $1$ & $1$ & $1$ & $1.012$ & $1.001$  \\
Skew. & $1.480,1.834$ & $0$ & $1.487,1.917$ & $1.67,2.124$ & $1.62,2.197$ & $1.138,1.368$ & $0$  &$1.163,1.368$ & $1.159,1.352$ \\
 & $1.643$ & $0$ & $1.68$ & $1.97$ & $1.858$  & $1.245$ & $0$  & $1.25$ & $1.253$ \\
Ex. kurt. & $3.123,7.222$ & $0$ & $3.148,6,753$ & $3.690,8.31$ & $3.5,9.428$ & $1.374,2.853$ & $0$ & $1.392,2.707$ & $1.481,2.88$ \\
 & $4.639$ & $0$ & $4.575$ & $6.377$ & $5.435$ & $1.95$ & $0$ & $1.989$ & $2$ \\
$q(0.025)$ & $-1.235,-1.194$ & $-1.96$& $-1.206,-1.174$ & $1.115,1.087$& $-1.114,-1.091$ & $-1.257,-1.226$ & $-1.96$ & $-1.276,-1.239$ & $-1.266,-1.23$ \\
 & $-1.215$ & $-1.96$& $-1.19$ & $-1.1$& $-1.101$ & $-1.245$ & $-1.96$ & $-1.257$ & $-1.246$ \\
$q(0.05)$ & $-1.15,-1.104$ & $-1.644$& $-1.115,-1.085$ & $-1.048,-1.023$& $-1.047,-1.024$ & $-1.18,-1.146$ & $-1.644$ & $-1.184,-1.154$ & $-1.175,-1.146$  \\
 & $-1.123$ & $-1.644$& $-1.1$ & $-1.033$& $-1.036$ & $-1.162$ & $-1.644$ & $-1.17$ & $-1.161$  \\
$q(0.95)$ & $1.861,2.07$ & $1.644$& $1.82,2.013$ & $1.863,2.081$& $1.873,2.066$ & $1.844,2.024$ & $1.644$ & $1.883,2.066$ & $1.856,2.034$ \\
& $1.946$ & $1.644$& $1.914$ & $1.942$& $1.969$ & $1.949$ & $1.644$ & $1.958$ & $1.949$ \\
$q(0.975)$ & $2.436,2.735$ & $1.96$ & $2.436,2.732$& $2.434,2.823$& $2.536,2.792$ & $2.328,2.607$  & $1.96$ & $2.383,2.645$ & $2.360,2.62$  \\
& $2.587$ & $1.96$ & $2.549$& $2.642$& $2.634$ & $2.486$  & $1.96$ & $2.5$ & $2.488$  \\
RMSE & $0.053$ & $0.762$ & $0.062$ & $0.11$ & $0.108$ & $0.040$ & $0.663$ & $0.048$ & $0.040$ \\
\hline
\end{tabular}
}

\caption{
Simulations based on $100$ independent repeats of $10000$ independent 
draws of each random variable (population size for Alg. \ref{algFsimPop}) 
i.e. columns, bar $\mathcal{N}(0,1)$. The value on the left is the minimum observed from the $100$ repeats,
on the right the maximum and in the line below from pooling all repeats together.
The running times are averages of $100$ independent repeats with $10000$ draws each.
The abbreviations in the row names are for average (Avg.), variance (Var.), skewness (Skew.) 
excess kurtosis (Ex. kurt.) and root--mean--square--error (RMSE).
The rows $q(\alpha)$ correspond to the, 
simulated, bar $\mathcal{N}(0,1)$, quantiles i.e. for a random variable $X$, $P(X \le q(\alpha)) = \alpha$. 
All simulations were done in R with the package TreeSim \citep{TreeSim1,TreeSim2} used to obtain the Yule trees
with speciation rate $\lambda=1$, $n=500$ tips and a root edge. 
The Yule tree $\Phi^{(n)}$, $\tilde{\Phi}^{(n)}$ values are centred and scaled by expectation and standard 
deviation from Eqs. \eqref{eqEvarPhiCont} and \eqref{eqEvarPhiDisc}. 
Other centrings and scalings are summarized in Tab. \ref{tabCIZtr}.
$N=10$ for Algs. \ref{algFsimPop} and \ref{algFsimRecurs} is the number of generations  
and recursion depth of the respective algorithm.
In Alg. \ref{algFsimPop} the initial population is set at $0$ and also $Y_{0}=0$ for Alg. \ref{algFsimRecurs}.
The simulations were run in R $3.4.2$ for openSUSE $42.3$ (x$86$\_$64$) 
on a $3.50$GHz. Intel\textregistered Xeon\textregistered CPU E$5$--$1620$ v$4$. 
The calculation of the RMSE is described in the text next to Eq. \eqref{eqRMSE}.
}\label{tabSimQuantResContDisc}
\end{sidewaystable}

\begin{table}[!htp]
\centering
\begin{tabular}{cccc}
\hline
&&&\\
& $\overline{W}_{N}$ & $Y^{(N)}$ & $\tilde{Y}^{(N)}$  \\
\hline
Centring ($\mu$) & $2(n-H_{n,1})/(n-1)$ & $1$ & $1$  \\ 
Scaling ($\sigma$) & $\sqrt{(2\pi^{2}-9)/9}$ & $\sqrt{1+1/16}$ & $(\sqrt{12})^{-1}$ \\
\hline
&&&\\
&  $\overline{W}_{N}^{NRE}$ & $Y^{(N)}_{NRE}$ & $\tilde{Y}^{(N)}_{NRE}$ \\
\hline
Centring ($\mu$) &  $(n+1-2H_{n,1})/(n-1)$ & $0$ & $0$  \\ 
Scaling ($\sigma$) & $\sqrt{(2\pi^{2}-18)/9}$ & $1/4$& $(\sqrt{12})^{-1}$ \\
\hline
\end{tabular}
\caption{Centrings and scalings applied to obtain the entries in Tab. \ref{tabSimQuantResContDisc}.
For a random variable $X$ by its centred and scaled version we mean $(X-\mu)/\sigma$. These
centrings and scalings are required to obtain mean zero, variance $1$ versions of the random variables,
i.e. so that they have the same location and scale as the $z$--transformed cophenetic index. In 
case of $\overline{W}_{N}$ we take the asymptotic scaling (Thm. \ref{thmVarWn})
to be comparable with $Y^{(N)}$. 
For the convenience of the reader we also provide corresponding centrings and scalings
in the no root--edge setup (not considered in Tab. \ref{tabSimQuantResContDisc}).
}
\label{tabCIZtr}
\end{table}

\begin{figure}[!htp]
\centering
\includegraphics[width=0.47\textwidth]{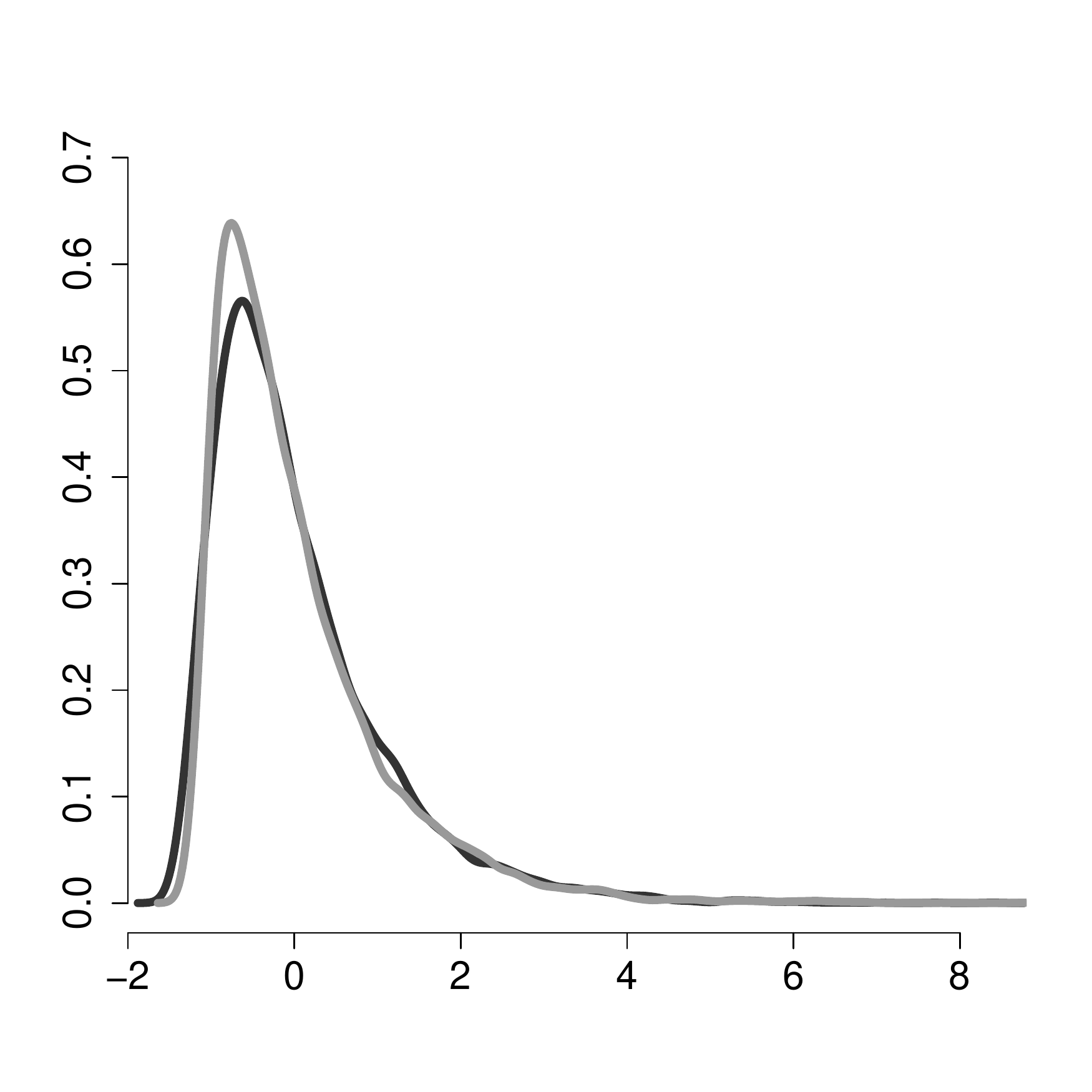}
\includegraphics[width=0.47\textwidth]{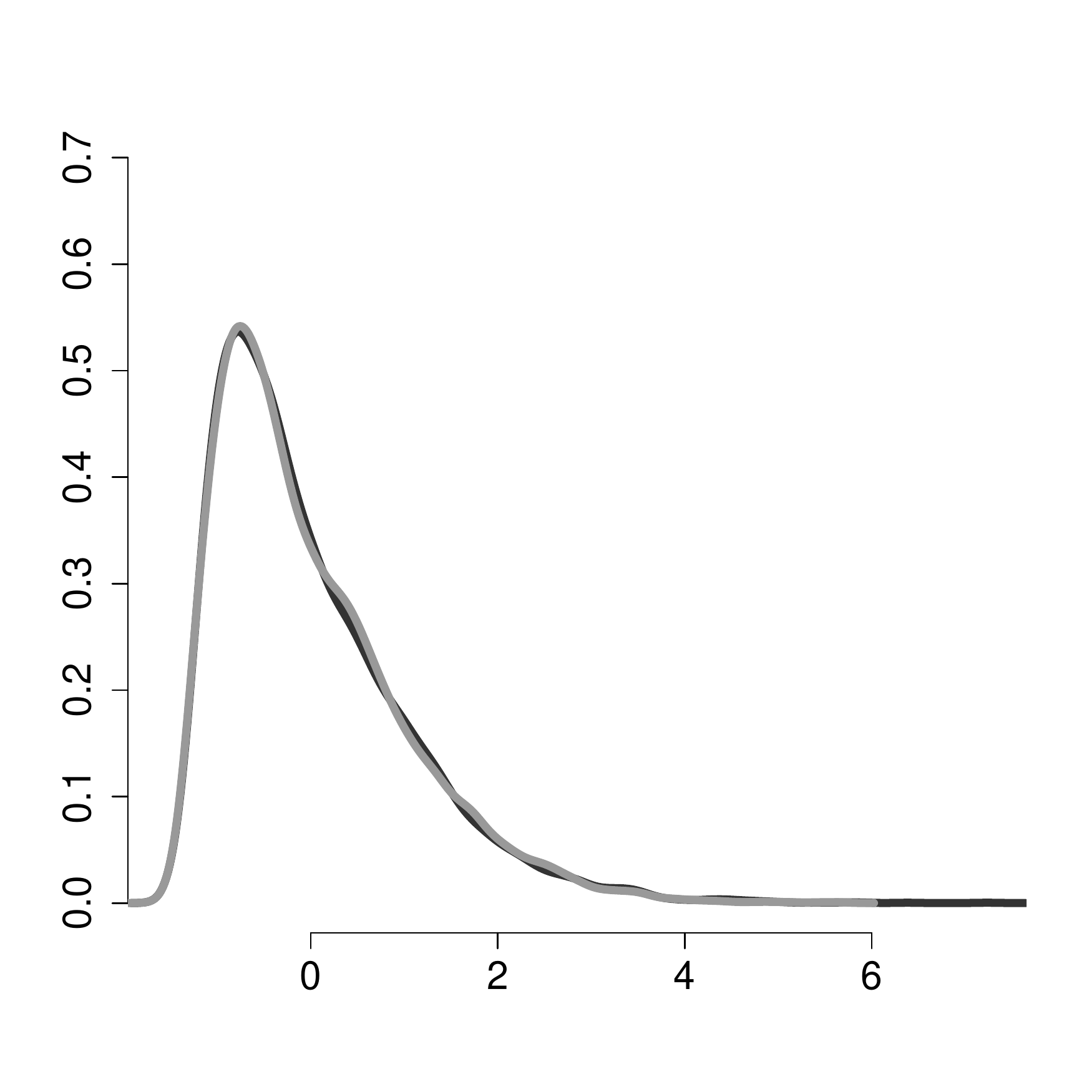}
\caption{
Density estimates of scaled (by theoretical standard deviation) 
and centred (by theoretical expectation) cophenetic indices (black) from $10000$ simulated $500$ tip Yule trees with $\lambda=1$
and of simulation by Alg. \ref{algFsimRecurs} (gray), also scaled and centred to mean $0$ and variance $1$.
Left: density estimates for $\Phi^{(n)}$, right: $\tilde{\Phi}^{(n)}$.
The curves are calculate by R's \texttt{density()} function. 
\label{figSimulRecurs}
}
\end{figure}

\subsection{Power of the tests}\label{sbsecPower}
For a given test statistic to be useful one also needs to know its power, the ability to reject the
null hypothesis (here Yule tree) when a given alternative one is true. For
example, balance indices based only on topology like Sackin's, Colless' or $\tilde{\Phi}^{(n)}$
cannot be expected to differentiate between any trees that
are generated by different constant rate birth--death processes or by the coalescent.
The rationale behind this is that the topologies induced by the $n$ contemporary
species (i.e. we forget about lineages leading to extinct ones) are stochastically
indistinguishable no matter what the death or birth rate is \citep[Thm. 2.3, Cor. 2.4][]{TGer2008a}.
Similarly, regarding the coalescent at the bottom of their  p. $93$ \citet{MSteAMcK2001} write
``$\ldots$, one has the coalescent model [1,18,19]. In this model one starts with
$n$ objects, then picks two at random to coalesce, giving $n-1$ objects. This process is repeated until there is
only a single object left. If this process is reversed, starting with one object to give
$n$ objects, then it is equivalent to the Yule model. Note that in the coalescent model there is commonly a probability
distribution for the times of coalescences, but in the Yule model we ignore this element.''
To differentiate between such trees one needs to take into consideration the branch lengths. 
Here we compare the power of the Sackin's, Colless', $\Phi^{(n)}$ and  $\tilde{\Phi}^{(n)}$
indices at the $5\%$ significance level. 

The null hypothesis is always that the tree is generated by a pure--birth process
with rate $\lambda=1$. The alternative ones are birth--death processes $(\lambda=1$, death rate $\mu=0.25$, and $0.5$
using the TreeSim package$)$,
coalescent process \citep[ape's \texttt{rcoal()} function][]{EParJClaKStr2004} and 
the biased speciation model for 

$$p\in \{0.05,0.1,0.125,0.15,0.18,0.2,0.25,0.4,0.5 \}.$$
We also simulate a pure--birth process to check if the significance level is met. All trees were simulated
with an $\exp(1)$ root edge.

The so--called biased speciation model with parameter $p$ is the tree growth model as
described by \citet{MBluOFra2005}. In their words,
``Assume that the speciation rate of a specific
lineage is equal to $r$ $(0 \le r \le 1)$. When a species with speciation rate $r$ splits, one of its descendant
species is given the rate $pr$ and the other is given the speciation rate $(1-  p)r$ where $p$ is fixed for the
entire tree. These rates are effective until the daughter species themselves speciate. Values of $p$
close to $0$ or $1$ yield very imbalanced trees while values around $0.5$ lead to over--balanced
phylogenies.'' We simulated such trees with in--house R code.

The quantiles of Sackin's and Colless' indices were obtained using Alg. \ref{algFsimRecurs}. 
It is known \citep[Eqs. 2 and 3][]{MBluOFra2005,MBluOFraSJan2006} that after normalization 
(centring by expectation and dividing by $n$) in the limit
they satisfy a contraction--type distribution of the form of Eq. \eqref{eqContrDist}, i.e.

$$
Y \stackrel{\mathcal{D}}{=} \tau Y' + (1-\tau)Y'' + C(\tau)
$$
for $\tau\sim\mathrm{Unif}[0,1]$. 
The function $C(\tau)$ takes the form

$$C(\tau) = 2\tau\log \tau+2(1-\tau)\log(1-\tau)+1$$ 
in Sackin's case and

$$C(\tau) = \tau\log \tau+(1-\tau)\log(1-\tau)+1-2\min(\tau,1-\tau)$$
in Colless' case. It particular, studying the limit of Sackin's index is equivalent
to studying the Quicksort distribution \citep{MBluOFra2005}.
We can immediately see
the main qualitative difference, the limit of the normalized cophenetic index has
the square in $\tau^{2}$, $(1-\tau)^{2}$ in the ``recursion part'' while Sackin's and Colless'
have $\tau$, $(1-\tau)$. 

Using $10000$ repeats of Alg. \ref{algFsimRecurs} with recursion depth $10$ we obtained the following
sets of quantiles $q(0.025)=-0.983$, $q(0.95)=1.189$,  $q(0.975)=1.493$, 
and $q(0.025)=-1.354$, $q(0.95)=1.494$,  $q(0.975)=1.868$ respectively for the normalized Sackin's and Colless' indices.

Under each model we simulated $10000$ trees conditioned on $500$ contemporary tips.
We then checked if the tree was outside the $95\%$ ``Yule zone'' \citep{GYanPAgaGYed2017}
by the procedure described in Alg. \ref{algSignifTest}.
We calculated the normalized Sackin's, Colless', discrete and
continuous cophenetic indices (normalizations from Thm. \ref{thmEVarCI}).
The functions \texttt{sackin.phylo()} and \texttt{colless.phylo()} of the phyloTop \citep{MKenMBoyCCol2016} R
package were used while the two cophenetic indices were calculated using a linear time in--house
R implementation based on traversing the tree \citep[Cor. 3][]{AMirFRosLRot2013}.
Two tests were considered, a two--sided one and a right--tailed one. 
For the discrete cophenetic index the quantiles from the simulation by 
Alg. \ref{algFsimRecurs} were considered, for the continuous those from $\overline{W}_{N}$ (Tab. \ref{tabSimQuantResContDisc}).
The power is then estimated as the fraction of times the null hypothesis was rejected
and represented in Tab. \ref{tabCIpower} by the corresponding Type II error rates.
For the Yule tree simulation we can see that the significance level is met. 
All simulated trees are independent of the trees used to obtain 
the values in Tab. \ref{tabSimQuantResContDisc} and quantiles of Sackin's and Colless' indices.
Hence, they offer a validation of the rejection regions. We summarize the power study in 
Tab. \ref{tabCIpower}.

\begin{sidewaystable}[!htp]
\centering
\begin{tabular}{lcccccccccccc}
\hline
&&&&&&&&\\
Model & \multicolumn{2}{c}{Sackin's} & \multicolumn{2}{c}{Colless'}
& \multicolumn{2}{c}{
$\tilde{\Phi}$} & \multicolumn{4}{c}{
$\Phi$ } & & \\
& $>$ & $2$ & $>$ & $2$ 
& $>$ & $2$ & $>$ & $>^{c}$ 
& $2$ & $2^{c}$ & mean($\hat{\lambda}$) & variance($\hat{\lambda}$) \\
\hline
Yule & $0.952$ & $0.952$ & $0.955$ & $0.955$ & $0.953$ & $0.952$ & $0.949$ &$0.95$ & $0.944$ & $0.944$ & $1$ & $0.002$ \\
Coalescent & $0.955$ & $0.954$ & $0.956$ & $0.959$ & $0.952$ & $0.955$ & $0.936$& $0$ & ${0.881}$ & $0$ & $37.836$ & $42.06$\\
birth--death $\mu=0.25$ & $0.952$ & $0.953$ & $0.956$ & $0.955$ & $0.948$ & $0.952$ & ${0.853}$ & $0.874$ & $0.903$ & $0.91$ & $0.87$ & $0.002$ \\
birth--death $\mu=0.5$ & $0.95$ & $0.95$ & $0.952$ & $0.955$ & $0.951$ & $0.953$ & ${0.635}$ & $0.729$& $0.739$&$0.808$ & $0.722$& $0.001$ \\
biased speciation $p=0.05$ & $0$ & $0$ & $0$ & $0$ & $0$ & $0$ & $0$ &$0.98$&  $0$ &$0.542$& $0.004$& $4.213\cdot 10^{-8}$ \\
biased speciation $p=0.1$ & $0$ & $0$ & $0$ & $0$ & $0$ & $0$ & $0$ &$0.982$&  $0$&$0.521$ &$0.004$ & $4.241\cdot 10^{-8}$ \\
biased speciation $p=0.125$ & $0$ & $0$ & $0$ & $0$ & $0.016$ & $0.431$ & $0$ &$0.981$&  $0$ & $0.522$&$0.004$ &$4.211\cdot 10^{-8}$ \\
biased speciation $p=0.18$ & $1$ & $1$ & $0.497$ & $0.959$ & $1$ & $1$ & ${0}$ & $0.981$& ${0}$ &$0.524$ &$0.004$ &$4.191\cdot 10^{-8}$ \\
biased speciation $p=0.2$ & $1$ & $1$ & $1$ & $1$ & $1$ & $1$ & ${0}$ &$0.982$&  ${0}$ & $0.508$ &$0.004$ & $4.222\cdot 10^{-8}$\\
biased speciation $p=0.25$ & $1$ & $0.834$ & $1$ & $1$ & $1$ & $1$ & ${0}$ &$0.98$ & ${0}$ &$0.515$ &$0.004$ & $4.243\cdot 10^{-8}$\\
biased speciation $p=0.4$ & $1$ & ${0}$ & $1$ & ${0}$ & $1$ & ${0}$ & ${0}$ &$0.982$ & $0.001$ &$0.51$&$0.004$ & $4.218 \cdot 10^{-8}$  \\
biased speciation $p=0.5$ & $1$ & ${0}$ & $1$ & ${0}$ & $1$ & ${0}$ & ${0}$ &$0.983$ & $0.002$ &$0.509$&$0.004$ & $4.335 \cdot 10^{-8}$ \\
\hline
\end{tabular}
\caption{
Power, presented as Type II error rates, of the various indices to detect deviations from the Yule tree for various alternative models
at the $5\%$ significance level. In the first row the trees were simulates under the Yule (i.e. we present the Type I error rate)
so this is a confirmation of correct significance level. Each probability is the fraction of $10000$ independently
simulated trees that were accepted as Yule trees by the various tests. Columns with ``$>$'' label
indicate right--tailed test and with label ``$2$'' the two--sided test. 
The critical regions for the cophenetic indices were taken from the pooled estimates in Tab. \ref{tabSimQuantResContDisc}.
The superscript $c$ indicates tests, where the trees were corrected for the speciation rate
through multiplying all branch lengths with $\hat{\lambda}$. 
The mean and variance over all trees 
of $\hat{\lambda}$, as obtained through ape's \texttt{yule()} function is reported.
Each tree's branches were scaled by its particular $\hat{\lambda}$ estimate.
\label{tabCIpower}
}
\end{sidewaystable}

As indicated in Alg. \ref{algSignifTest} one should first ``correct'' the tree for the speciation rate,
when using the cophenetic index with branch lengths.
The distributional results derived here on the cophenetic indices are for a unit speciation rate
Yule tree. For a mathematical perspective this is not a significant restriction.
If one has a pure--birth tree generated by a process with speciation rate $\lambda\neq 1$, then
multiplying all branch lengths by $\lambda$ will make the tree equivalent to one with unit rate. 
Hence, all the results presented here are general up to a multiplicative constant. 
However, from an applied perspective the situation cannot be treated so lightly. 
For example, if we used the cophenetic index with branch lengths from a Yule tree with a very large
speciation rate, then we would expect a significant deviation. However, unless one is 
interested in deviations from the unit speciation rate Yule tree, this would not be useful.
Hence, one needs to correct for this effect. If the tree did come from a Yule process,
then an estimate, $\hat{\lambda}$, of the speciation rate by maximum likelihood can be obtained.
For example, in the work here we used ape's \texttt{yule()} function. Then, one multiplies 
all branch lengths in the tree by $\hat{\lambda}$ and calculates the cophenetic index for 
this transformed tree. It is important to point out that $\hat{\lambda}$ is only an estimate
and hence a random variable. The effects of the this source of randomness on the limit
distribution deserve a separate, detailed study. 
Balance indices that do not use branch lengths do not suffer from this issue but on the other
hand miss another aspect of the tree---proportions between branch lengths that are
non--Yule like.

The power analysis presented in Tab. \ref{tabCIpower} generally agrees with intuition and the power
analysis done by \citet{MBluOFra2005}. The first row shows that for all tests and statistics
the $5\%$ significance level is approximately kept. Then, in the next three rows (coalescent and
birth--death process) all topology based indices fail completely (the power is at the significance level).
This is completely unsurprising as the after one removes all speciation events (with lineages) leading
to extinct species from a birth--death tree, the remaining tree is topologically equivalent to a pure
birth tree. The same is true for the coalescent model, its topology is identical
in law to the Yule tree's one.
The cophenetic index with branch lengths has a high Type II error rate but is still better, than the
topological indices. However, when one ``corrects for $\lambda$'' this index manages to nearly ($2$ trees were
not rejected by the two--sided test) perfectly reject the coalescent model trees.

Power for the biased speciation model follows the same pattern as \citet{MBluOFra2005}
observed. When imbalance is evident, $p\le 0.125$, all ($\lambda$ uncorrected) 
tests were nearly perfect (two--sided discrete cophenetic is an exception). 
However, the $\lambda$ correction significantly 
worsened the ability of $\Phi^{(n)}$ to detect deviations.
As imbalance decreased so did the power of the topological indices. 
For overbalanced trees one--sided 
tests failed, two sided worked \citep[just as][observed]{MBluOFra2005}. The cophenetic index with branch
lengths (without correction), that does not consider only the topology, was able to successfully reject the pure--birth tree for all $p$
(with only minimal Type II error for $p\ge 0.4$ in the two--sided test case). 
Interestingly, $\Phi^{(n)}$'s (both corrected and uncorrected) power seems invariant with respect to $p$.
These results are especially promising as $\Phi^{(n)}$ seems to be an index that
functions significantly better in the difficult, $0.18\le p \le 0.25$, regime, even after correcting.

At this stage we can point out that a normal approximation to the cophenetic indices' limit distribution is not appropriate.
When doing the above power study 
we observed that when using the quantiles of the standard normal distribution
the right--tailed test based on 
$\Phi^{(n)}$ rejects $6.81\%$ of Yule trees, based on $\tilde{\Phi}^{(n)}$ rejects $7.03\%$ of Yule trees,
two--sided $\Phi^{(n)}$ test rejects $4.87\%$ of Yule trees and
two--sided $\tilde{\Phi}^{(n)}$ test rejects $4.66\%$ of Yule trees.
The Type I error rates of the two--sided tests are within the observed Monte Carlo errors (in Tab. \ref{tabCIpower}) 
but the right--tailed tests' Type I error are evidently inflated. 
This confirms that the right tail of the scaled cophenetic index is much heavier than normal.

In short the power study indicates that the cophenetic index with branch lengths should be considered as an
option to detect deviations from the Yule tree. This is because it is able 
to use information from two sources---the topology and time
\citep[a needed direction of development, as indicated in Ch. $33$ of][]{JFel2004}. Actually, this is evident
in the decomposition of Eq. \eqref{eqWnZi}. The $V_{i}^{(n)}$s describe the topology 
and the $Z_{i}$s branch lengths. With more information a more powerful testing procedure
is possible. Deviations that are not topologically visible, e.g. biased speciation
in the $0.18 \le p \le 0.25$ regimes, are now detectable. To use $\Phi^{(n)}$ one should correct for the effects of 
the speciation rate, as otherwise one merely detects deviations from the unit rate Yule tree. This 
correction is a mixed blessing. It can help or hinder detection. 

\subsection{Examples with empirical phylogenies}\label{sbsecRealPhyl}
It is naturally interesting to ask how do the indices behave for phylogenies estimated from sequence data.
Comparing a database of phylogenies, like TreeBase (\url{http://www.treebase.org}), with yet another
index's distribution under the Yule model should not be expected to yield interesting
results. The Yule model has been indicated as inadequate to describe the collection of TreeBase's trees
\citep[e.g.][]{MBluOFra2007}. Therefore, we choose a particular study that estimated a tree and also reported a 
collection of posterior trees.
\citet{VSosJOrnSRamEGan2016} is a recent work, 
providing all trees from BEAST's \citep{ADruARam2007} output,
well suited for such a purpose. 
\citet{VSosJOrnSRamEGan2016} estimate the evolutionary relationships between a set of $109$ tree ferns species.
They report a posterior set of $22498$ phylogenies \citep{VSosJOrnSRamEGan2016dryad}. 

In Tab. \ref{tabCISOSAdata} we look what percentage of the trees from the posterior was accepted
as being consistent with the Yule tree by the various tests and indices. It can be seen
that the discrete cophenetic index has a high acceptance rate. The continuous one, which also takes
into account branch lengths did not accept a single tree. However, this is lost when one
corrects for the speciation rate (first ape's \texttt{multi2di()} was used to make the
trees binary ones). 
Most tests and indices rejected the Yule
tree for the maximum likelihood estimate of the phylogeny with some exceptions. The two--sided
discrete cophenetic index test did not reject the null hypothesis of the pure--birth tree.
Also after correcting for the speciation rate (estimated at $\hat{\lambda}=0.023$),
neither test based on the continuous cophenetic index rejected the Yule tree.
Therefore, one should conclude (based on the ``topological balances'') 
that the Yule tree null hypothesis can be rejected for this clade of plants.

\begin{table}[!htp]
\centering
\begin{tabular}{cccccccccc}
\hline
&&&&&\\
\multicolumn{2}{c}{Sackin's}  & \multicolumn{2}{c}{Colless'}
& \multicolumn{2}{c}{
$\tilde{\Phi}$} &  
\multicolumn{4}{c}{
$\Phi$}\\  
$>$ & Sackin's $2$ & $>$ & $2$ 
& $>$ & $2$ & $>$ & $>^{c}$ 
& $2$ & $2^{c}$ \\ 
\hline
$0.03$ & $0.109$ & $0.019$ & $0.062$ & $0.385$ & $0.559$ & $0$ & $0.974$ & $0$ & $0.995$ \\ 
\hline
\end{tabular}
\caption{Percentage of trees from 
\citet{VSosJOrnSRamEGan2016dryad}'s set of posterior trees accepted as Yule trees by the various tests
and indices. Columns with ``$>$'' label
indicate right--tailed test and with label ``$2$'' the two--sided test.
The critical regions for the cophenetic indices were taken from the pooled estimates in Tab. \ref{tabSimQuantResContDisc}.
The superscript $c$ indicates tests, where the trees were corrected for the speciation rate
through multiplying all branch lengths with $\hat{\lambda}$. 
Ape's \texttt{yule()} function returned an average over all trees estimate of $\lambda$
of $0.023$ with variance $2.988 \cdot 10^{-6}$. Each tree's branches were scaled by 
its particular $\hat{\lambda}$ estimate. Each tree was first transformed by
ape's \texttt{multi2di()} into a binary one.
\label{tabCISOSAdata}
}
\end{table}

We also followed \citet{MBluOFra2005} in looking at 
\citet{KYusetal2001}'s
phylogeny of the human immunodeficiency virus type 1 (HIV--1) group M gene sequences, available in the ape R package.
The phylogeny consists of $193$ tips and \citet{MBluOFra2005} could not reject the null hypothesis
of the pure--birth tree (using Sackin's index amongst others). After pruning the tree 
to keep ``only the old internal branches that corresponded to the $30$ oldest ancestors'' they
were able to reject the Yule tree. They conclude that the
``results probably
indicate a change in the evolutionary rate during the evolution which had more impact on
cladogenesis during the early expansion of the virus.''
Repeating their experiment we find that only the two versions
of the cophenetic index point to a deviation but only in the two--sided test (see Tab. \ref{tabTestHIVtree}). 
Based only on the $\tilde{\Phi}^{(n)}$'s test and that it conflicts with the conclusions of Sackin's and Colless'
one should not draw any conclusions. However, as $\Phi^{(n)}$'s test indicates a deviation, we can be 
inclined to reject the null hypothesis of the Yule tree. This is further strengthened by the fact
that the significance remains after the correction for $\lambda$. Even though the topology as a whole seems consistent
with the pure--birth tree the branch lengths are not. 
The fact that only the two--sided test rejected the Yule tree indicates that the
HIV phylogeny is over--balanced in comparison to a pure--birth tree.
In fact, in the biased speciation model tree over--balance is observed for values of $p$ close to $0.5$ \citep{MBluOFra2005}.
Such trees have a declining speciation rate as they grow and hence this supports \citet{MBluOFra2005}'s aforementioned
explanation.

\begin{table}[!htp]
\centering
\begin{tabular}{cccccc}
\hline
&&&&&\\
Sackin's & Colless' & $\tilde{\Phi}$  &  $\Phi$ & $\Phi^{c}$  & $\hat{\lambda}$\\
\hline
 $0.823^{-,-}$ & $0.993^{-,-}$ & $-1.689^{-,\ast}$ & $-1.765^{-,\ast}$ 
 & $-1.602^{-,\ast}$ & $9.313$ \\
 \hline
\end{tabular}
\caption{Values of the normalized indices for 
\citet{KYusetal2001}'s HIV--$1$ phylogeny. Above each index is an indication if the
index deviates at the $5\%$ significance level from the Yule tree, dash insignificant,
asterisk significant. The first symbol concerns the right--tailed test, the second the two--sided test.
The superscripted $\Phi^{c}$ is calculated from the tree corrected for the speciation rate
by multiplying all branch lengths by $\hat{\lambda}$.
\label{tabTestHIVtree}
}
\end{table}

\section{Almost sure behaviour of the cophenetic index}\label{secASbehaviour}
We study the asymptotic distributional properties of $\Phi^{(n)}$
for the pure--birth tree model using techniques from
our previous papers on branching Brownian and 
Ornstein--Uhlenbeck processes \citep{KBar2014,KBarSSag2014,KBarSSag2015,SSagKBar2012}.
We assume that the speciation rate of the tree is $\lambda=1$.
The key property we will use is that in the pure--birth tree case
the time between two speciation events, $k$ and $k+1$ (the first speciation
event is at the root), is $\exp(k)$ distributed, 
as the minimum of $k$ $\exp(1)$ random variables. 
We furthermore, assume that the tree starts with a single species (the origin)
that lives for $\exp(1)$ time and then splits (the root of the tree) into two species. 
We consider a conditioned on $n$ contemporary species tree. This conditioning
translates into stopping the tree process just before the $n+1$ speciation event,
i.e. the last interspeciation time is $\exp(n)$ distributed.
We introduce the notation that $U^{(n)}$ is the height of the tree,
$\tau^{(n)}$ is the time to coalescent of two randomly selected tip species
and $T_{k}$ is the time between speciation events $k$ and $k+1$ 
\citep[see Fig. \ref{figTreeTimes} and][]{KBarSSag2015,SSagKBar2012}.

\begin{figure}[!htp]
\centering
\includegraphics[width=0.6\textwidth]{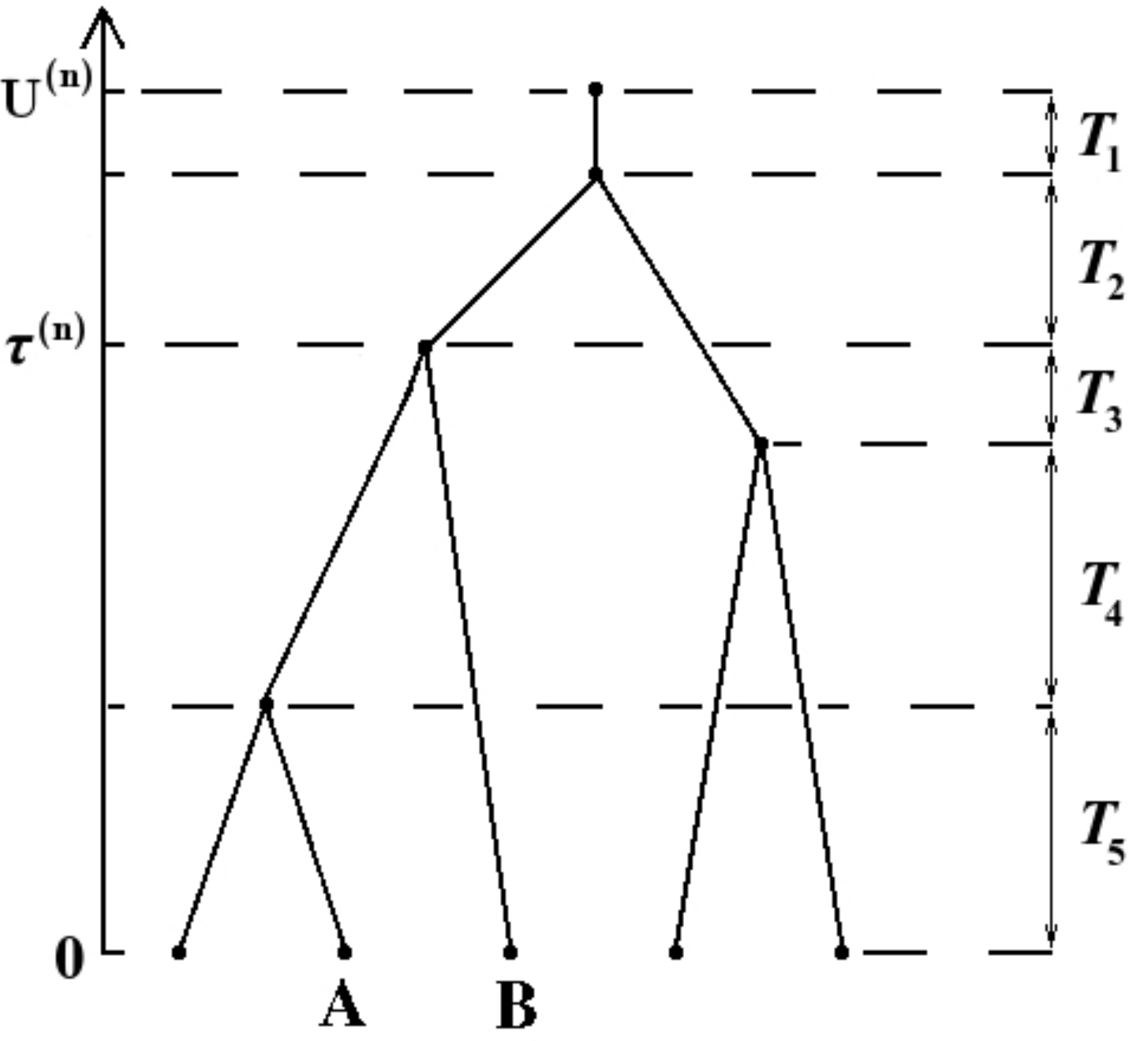}
\caption{
A pure--birth tree with the various time components marked on it. 
The between speciation times on this lineage are
$T_{1}$, $T_{2}$, $T_{3}+T_{4}$ and $T_{5}$.
If we ``randomly sample'' the pair of extant 
species ``A'' and ``B'', then the two nodes coalesced at time $\tau^{(n)}$.
\label{figTreeTimes}
}
\end{figure}

\begin{thm}\label{thmRecursCI}
The cophenetic index is an increasing sequence of random variables, $\Phi^{(n+1)} >\Phi^{(n)}$
and has the recursive representation

\be
\Phi^{(n+1)} = \Phi^{(n)} + nU^{(n)} - \sum\limits_{i=1}^{n}\xi^{(n)}_{i}\sum\limits_{i\neq j}^{n}\tau^{(n)}_{ij},
\ee
where $\xi^{(n)}_{i}$ is an indicator random variable whether tip $i$ split at the $n$--th speciation event.
\end{thm}

\begin{proof}
From the definition we can see that 

$$
\Phi^{(n)} = \sum\limits_{1 \le i < j \le n}\left(U^{(n)}-\tau^{(n)}_{ij}\right)
= \binom{n}{2}\left(U^{(n)}-\E{\tau^{(n)} \vert \mathcal{Y}_{n}}\right),
$$
where $\tau^{(n)}_{ij}$ is the time to coalescent of tip species $i$ and $j$.
We now develop a recursive representation for the cophenetic index.
First notice that when a new speciation occurs all coalescent times are extended
by $T_{n+1}$, i.e.

$$
\begin{array}{l}
\sum\limits_{1 \le i < j \le n+1}\tau^{(n+1)}_{ij}
= 
\sum\limits_{1 \le i < j \le n}\left(\tau^{(n)}_{ij}+T_{n+1}\right)
+ \sum\limits_{i=1}^{n}\xi^{(n)}_{i}\sum\limits_{i\neq j}^{n}\left(\tau^{(n)}_{ij}+T_{n+1}\right)
+ T_{n+1},
\end{array}
$$
where the ``lone'' $T_{n+1}$ is the time to coalescent of the two descendants of the
split tip. The vector $\left(\xi^{(n)}_{1},\ldots,\xi^{(n)}_{n}\right)$ consists of $n-1$ $0$s and
exactly one $1$ (a categorical distribution with $n$ categories all with equal probability). 
For each $i$ the marginal probability that $\xi^{(n)}_{i}$ is $1$ is $1/n$.
We rewrite 

$$
\begin{array}{l}
\sum\limits_{1 \le i < j \le n+1}\tau^{(n+1)}_{ij}
= 
\binom{n+1}{2} T_{n+1}
\sum\limits_{1 \le i < j \le n}\tau^{(n)}_{ij}
+ \sum\limits_{i=1}^{n}\xi^{(n)}_{i}\sum\limits_{i\neq j}^{n}\tau^{(n)}_{ij}
\end{array}
$$
and then obtain the recursive form 

$$
\begin{array}{rcl}
\Phi^{(n+1)} & = & \binom{n+1}{2}U^{(n)}+\binom{n+1}{2}T_{n+1}
-
\binom{n+1}{2} T_{n+1}
-\sum\limits_{1 \le i < j \le n}\left(\tau^{(n)}_{ij}+T_{n+1}\right)
\\ && - \sum\limits_{i=1}^{n}\xi^{(n)}_{i}\sum\limits_{i\neq j}^{n}\tau^{(n)}_{ij}
\\ &= &
\binom{n+1}{2}U^{(n)}
-\sum\limits_{1 \le i < j \le n}\left(\tau^{(n)}_{ij}+T_{n+1}\right)
- \sum\limits_{i=1}^{n}\xi^{(n)}_{i}\sum\limits_{i\neq j}^{n}\tau^{(n)}_{ij}
\\ &= &
\Phi^{(n)} + nU^{(n)} - \sum\limits_{i=1}^{n}\xi^{(n)}_{i}\sum\limits_{i\neq j}^{n}\tau^{(n)}_{ij}.
\end{array}
$$
Obviously, $\Phi^{(n+1)}>\Phi^{(n)}$.
\end{proof}

\begin{proof}\textit{of Theorem \ref{thmWnContConv}.}
Obviously

$$
W_{n+1} = \frac{n-1}{n+1}W_{n}+\frac{2}{n+1}U^{(n)}-\binom{n+1}{2}^{-1}\sum\limits_{i=1}^{n}\xi^{(n)}_{i}\sum\limits_{i\neq j}^{n}\tau^{(n)}_{ij}
$$
and

$$
\begin{array}{rcl}
\E{W_{n+1}\vert \mathcal{Y}_{n}} & = &
\frac{n-1}{n+1}W_{n}+\binom{n+1}{2}^{-1} 
\left(
nU^{(n)}
-\frac{1}{n}\sum\limits_{i=1}^{n}\sum\limits_{i\neq j}^{n}\tau^{(n)}_{ij}\right)
\\ &= &
\frac{n-1}{n+1}W_{n}+\binom{n+1}{2}^{-1} \frac{2}{n}
\left(
\frac{n^{2}}{2}U^{(n)}
-\sum\limits_{i<j}^{n}\tau^{(n)}_{ij}\right)
\\ &= &
\frac{n-1}{n+1}W_{n}+\binom{n+1}{2}^{-1} \frac{2}{n}
\left(\binom{n}{2}W_{n}+\frac{n}{2}U^{(n)}
\right)
\\ &= &
\left(\frac{n-1}{n+1}+\binom{n+1}{2}^{-1} \frac{2}{n}\binom{n}{2}\right)W_{n}+\binom{n+1}{2}^{-1} U^{(n)}
\\ &= &
\frac{(n-1)(n+2)}{n(n+1)}W_{n}+\binom{n+1}{2}^{-1} U^{(n)}
 \\ & = & W_{n} + \binom{n+1}{2}^{-1}(U^{(n)}-W_{n})
 = W_{n} + \binom{n+1}{2}^{-1}(U^{(n)}-\binom{n}{2}^{-1}\Phi^{(n)})
 \\ & =& W_{n} + \binom{n+1}{2}^{-1}\left(U^{(n)}-\binom{n}{2}^{-1}\binom{n}{2}(U^{(n)}-\E{\tau^{(n)}\vert \mathcal{Y}_{n}})\right)
 >W_{n}.
\end{array}
$$
Hence, $W_{n}$ is a positive submartingale with respect to $\mathcal{Y}_{n}$.
Notice that

$$
\Expectation{W_{n}^{2}}=\Expectation{(U^{(n)}-\EcYn{\tau^{(n)}})^{2}}  \le  \Expectation{(U^{(n)}-\tau^{(n)})^{2}}.
$$
Then, using the general formula for the moments of $U^{(n)}-\tau^{(n)}$
\citep[Appendix A,][]{KBarSSag2015}, we see that 

$$
\begin{array}{rcl}
\Expectation{(U^{(n)}-\tau^{(n)})^{2}} & = & 2\frac{n+1}{n-1}
\sum\limits_{j=1}^{n-1}\frac{1}{(j+1)(j+2)}\left(H_{j,1}^{2}+H_{j,2} \right)
\\&=& 2\frac{n+1}{n-1}\left(\frac{n}{n+1}H_{n,2}-\frac{n}{n+1}-\frac{H_{n,2}}{n+1}
+ \sum\limits_{j=1}^{n-1}\frac{H_{j,1}^{2}}{(j+1)(j+2)}\right)
\nearrow  \frac{2}{3}\pi^{2}.
\end{array}
$$
Hence, $\Expectation{W_{n}}$ and $\Expectation{W_{n}^{2}}$ are $O(1)$ and 
by the martingale convergence theorem $W_{n}$ converges almost surely and in $L^{2}$ to a finite
first and second moment random variable. 
\end{proof}

\begin{cor}
$W_{n}$ has finite third moment and is $L^{3}$ convergent.
\end{cor}

\begin{proof}
We first recall the $W_{n}$ is positive.
Using the general formula for the moments of $U^{(n)}-\tau^{(n)}$ again
we see 

$$
\begin{array}{rcl}
\Expectation{(U^{(n)}-\EcYn{\tau^{(n)}})^{3}} & \le & 
\Expectation{(U^{(n)}-\tau^{(n)})^{3}}
\\&=& 
2\frac{n+1}{n-1}\sum\limits_{j=1}^{n-1}\frac{1}{(j+1)(j+2)}\left(H_{j,1}+3H_{j,1}+3H_{j,2}+H_{j,3} \right)
\\ & < & 16\frac{n+1}{n-1}\sum\limits_{j=1}^{n-1}\frac{H_{j,1}}{(j+1)(j+2)}
\\ & = & 16\frac{n+1}{n-1}\frac{n-H_{n,1}}{n+1} = 16\frac{n-H_{n,1}}{n-1} \nearrow 16.
\end{array}
$$
This implies that $\E{W_{n}^{3}} = O(1)$ and hence
$L^{3}$ convergence and finiteness of the third moment.
\end{proof}

\begin{remark}
Notice that we \citep[Appendix A,][]{KBarSSag2015} made a typo
in the general formula for the cross moment of 

$$\E{(U^{(n)}-\tau^{(n)})^{m}\tau^{(n)^{r}}}.$$ 
The $(-1)^{m+r}$ 
should not be there, it 
will cancel with the $(-1)^{m+r}$ from the derivative of the Laplace transform.
\end{remark}

\begin{proof}\textit{of Theorem \ref{thmWnZi}.}
We write $W_{n}$ as

$$
\begin{array}{rcl}
W_{n} & = & U^{(n)} - \E{\tau^{(n)} \vert \mathcal{Y}_{n}} = 
\E{U^{(n)} -\tau^{(n)} \vert \mathcal{Y}_{n}}
=  \E{\sum\limits_{k=1}^{n-1}1^{(n)}_{k} \sum\limits_{i=1}^{k}T_{i} \vert \mathcal{Y}_{n}}
\\ & = & 
\E{\sum\limits_{i=1}^{n-1}T_{i}\sum\limits_{k=i}^{n-1}1^{(n)}_{k}  \vert \mathcal{Y}_{n}}
=
\sum\limits_{i=1}^{n-1}T_{i}\sum\limits_{k=i}^{n-1}\E{1^{(n)}_{k}  \vert \mathcal{Y}_{n}}
\\ & = & 
\sum\limits_{i=1}^{n-1}\left(\frac{1}{i}\sum\limits_{k=i}^{n-1}\E{1^{(n)}_{k}  \vert \mathcal{Y}_{n}}\right) Z_{i}
=\sum\limits_{i=1}^{n-1}V_{i}^{(n)}Z_{i},
\end{array}
$$
where $Z_{1},\ldots,Z_{n-1}$ are i.i.d. $\exp(1)$ random variables.
\end{proof}

\begin{remark}
We notice that we may equivalently rewrite

\be\label{eqWn.measure}
W_{n} = \sum\limits_{k=1}^{n-1}\EcYn{1^{(n)}_{k}} \left(\sum\limits_{i=1}^{k} T_{i} \right)
= \sum\limits_{k=1}^{n-1}\E{1^{(n)}_{k}  \vert \mathcal{Y}_{n}} \left(\sum\limits_{i=1}^{k} \frac{1}{i} Z_{i} \right).
\ee
The above and Eq. \eqref{eqWnZi} are very elegant representations of the cophenetic
index with branch lengths. They explicitly describe the way the cophenetic index
is constructed from a given tree.
\end{remark}

\begin{proof}\textit{of Theorem \ref{thmWnDiscConv}.}
The argumentation is analogous to the proof of Thm. \ref{thmWnContConv} by using
the recursion

$$
\tilde{\Phi}^{(n+1)} = \tilde{\Phi}^{(n)} + \sum\limits_{i=1}^{n}\xi_{i}^{(n)}
\left(\sum\limits_{i\neq j}^{n}\tilde{\phi}_{ij} + \Upsilon^{(n)}_{i}  \right),
$$
where $\Upsilon^{(n)}_{i}$ is the number of nodes on the path from the root (or appropriately origin) 
of the tree to tip $i$, \citep[see also][esp. Fig. A.$8$]{KBar2014}.
An alternative proof for almost sure convergence can be found in Section \ref{secDiffProc}.
\end{proof}

\section{Second order properties}\label{sec2ndOrd}
In this Section we prove a series of rather technical Lemmata and Theorems 
concerning the second order properties of $1_{k}^{(n)}$, $V_{i}^{(n)}$
and $W_{n}$. 
Even though we will not obtain any weak limit,
the derived properties do give insight on the delicate behaviour of $W_{n}$ and also show
that no ``simple'' limit, e.g. Eq. \eqref{eqASlimPhi}, is possible.
To obtain our results we used Mathematica $9.0$ for Linux x$86$ ($64$--bit) 
running on Ubuntu $12.04.5$ LTS to evaluate the required sums in closed forms. 
The Mathematica code is available as an appendix to this paper.

\begin{lemma}

\be
\variance{1_{k}^{(n)}} = 
2\frac{n+1}{n-1}\frac{1}{(k+1)(k+2)}\left(1 - 2\frac{n+1}{n-1}\frac{1}{(k+1)(k+2)} \right)
\ee
\end{lemma}

\begin{proof}

$$
\begin{array}{rcl}
\variance{1_{k}^{(n)}} & = & \Expectation{1_{k}^{(n)^{2}}} - \Expectation{1_{k}^{(n)}}^{2}
= \pi_{n,k} - \pi_{n,k}^{2} = \pi_{n,k}(1-\pi_{n,k}) 
\\ &=&
2\frac{n+1}{n-1}\frac{1}{(k+1)(k+2)}\left(1-2\frac{n+1}{n-1}\frac{1}{(k+1)(k+2)}\right).
\end{array}
$$
\end{proof}
The following lemma is an obvious consequence of the definition of $1_{k}^{(n)}$.

\begin{lemma}
For $k_{1} \neq k_{2}$ 

\be
\cov{1_{k_{1}}^{(n)}}{1_{k_{2}}^{(n)}} = -\pi_{n,k_{1}}\pi_{n,k_{2}}
=\frac{(-4)(n+1)^{2}}{(n-1)^{2}(k_{1}+1)(k_{1}+2)(k_{2}+1)(k_{2}+2)}.
\ee
\end{lemma}

\begin{lemma} \label{lemVar1k}

\be
\begin{array}{rcl}
\variance{\EcYn{1_{k}^{(n)}}} & = &
4\frac{n+1}{n(n-1)^{2}}\frac{(n-(k+1))(n(3k^{2}+5k-4)-(k^{2}-k-8))}{(k+1)^{2}(k+2)^{2}(k+3)(k+4)}.
\end{array}
\ee
\end{lemma}

\begin{proof}
Obviously

$$
\begin{array}{rcl}
\variance{\EcYn{1_{k}^{(n)}}} & = & \Expectation{\EcYn{1_{k}^{(n)}}^{2}} - \Expectation{\EcYn{1_{k}^{(n)}}}^{2}.
\end{array}
$$
We notice \citep[as][in Lemmata $11$ and $2$ respectively]{KBarSSag2015,KBar2016arXiv} that
we may write

$$
\Expectation{\EcYn{1_{k}^{(n)}}^{2}} = \Expectation{1_{k,1}^{(n)}1_{k,2}^{(n)}},
$$
where $1_{k,1}^{(n)}$, $1_{k,2}^{(n)}$ are two independent copies of $1_{k}^{(n)}$,
i.e. we sample a pair of tips twice and ask if both pairs coalesced at the $k$--th
speciation event. There are three possibilities, we (i) drew the same pair,
(ii) drew two pairs sharing a single node or (iii) drew two disjoint pairs.
Event (i) occurs with probability $\binom{n}{2}^{-1}$, (ii) with probability
$2(n-2)\binom{n}{2}^{-1}$ and (iii) with probability $\binom{n-2}{2}\binom{n}{2}^{-1}$.
As a check notice that $1+2(n-2)+\binom{n-2}{2}=\binom{n}{2}$. 
In case (i) $1_{k,1}^{(n)}=1_{k,2}^{(n)}$, hence writing informally

$$
\Expectation{1_{k,1}^{(n)}1_{k,2}^{(n)} \vert \mathrm{(i)}} = 
\Expectation{1_{k}^{(n)}} = \pi_{n,k}.
$$

\begin{figure}[!htp]
\begin{center}
\includegraphics[width=0.7\textwidth]{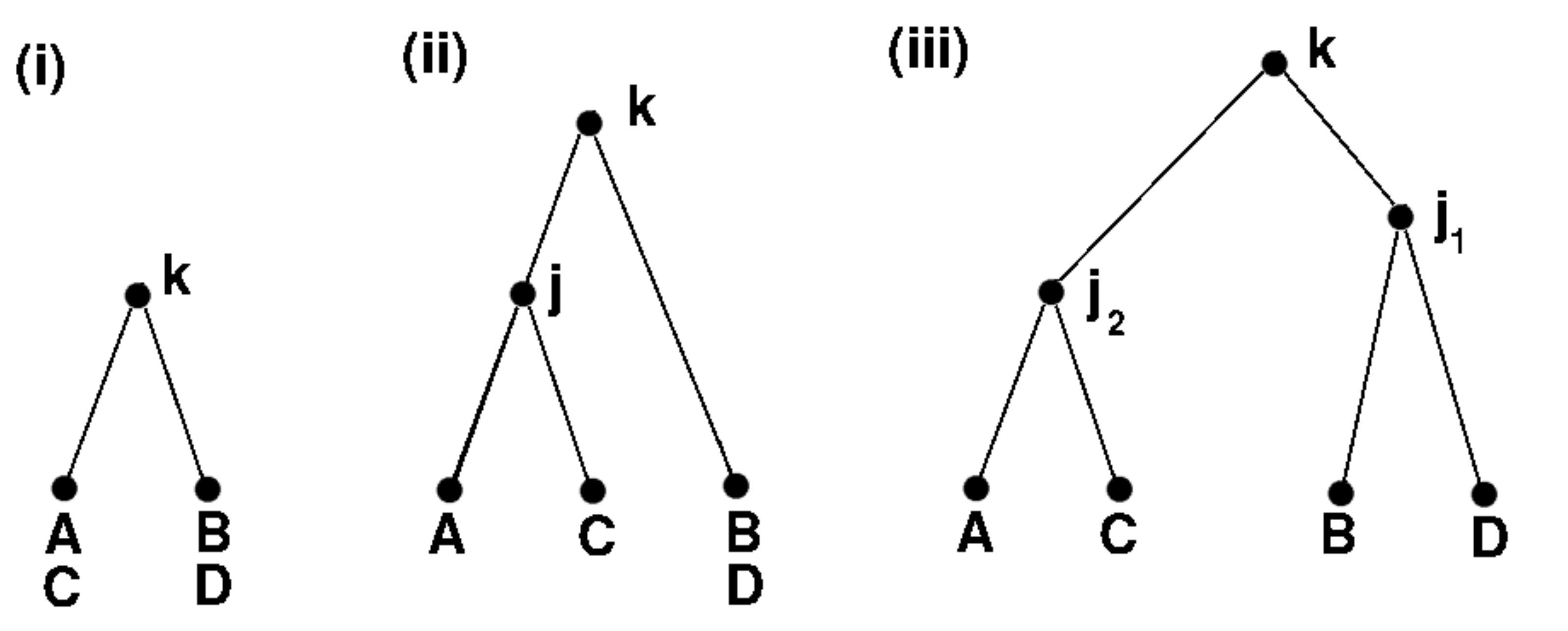}
\caption{The three possible cases when drawing two random pairs of tip species that
coalesce at the $k$--th speciation event. In the picture 
we ``randomly draw'' pairs $(A,B)$ and $(C,D)$. }\label{fig2pairphyl}
\end{center}
\end{figure}
To calculate cases (ii) and (iii) we visualize the 
situation in Fig. \ref{fig2pairphyl} and recall the proof of \citet{KBarSSag2015}'s Lemma $1$.
Using Mathematica we obtain

$$
\begin{array}{rcl}
\Expectation{1_{k,1}^{(n)}1_{k,2}^{(n)} \vert \mathrm{(ii)}}  & =&
\sum\limits_{j=k+1}^{n-1}
\left(1-\frac{3}{\binom{n}{2}}\right)\ldots\left(1-\frac{3}{\binom{j+2}{2}}\right)\frac{1}{\binom{j+1}{2}}
\left(1-\frac{1}{\binom{j}{2}}\right)\ldots
\\ && \cdot
\left(1-\frac{1}{\binom{k+2}{2}}\right)\frac{1}{\binom{k+1}{2}}
\\&=&
4\frac{(n+1)}{(n-1)(n-2)}\frac{n-(k+1)}{(1 + k) (2 + k) (3 + k)}.
\end{array}
$$
Similarly for case (iii)

$$
\begin{array}{rcl}
\Expectation{1_{k,1}^{(n)}1_{k,2}^{(n)} \vert \mathrm{(iii)}}  & =&
\sum\limits_{j_{2}=k+2}^{n-1}\sum\limits_{j_{1}=k+1}^{j_{2}+1}
\left(1-\frac{6}{\binom{n}{2}}\right)\ldots\left(1-\frac{6}{\binom{j_{2}+2}{2}}\right)\frac{4}{\binom{j_{2}+1}{2}}
\\&&\cdot
\left(1-\frac{3}{\binom{j_{2}}{2}}\right)\ldots\left(1-\frac{3}{\binom{j_{1}+2}{2}}\right)\frac{1}{\binom{j_{1}+1}{2}}
\left(1-\frac{1}{\binom{j_{1}}{2}}\right)\ldots
\\ && \cdot
\left(1-\frac{1}{\binom{k+2}{2}}\right)\frac{1}{\binom{k+1}{2}}
\\&=&
16\frac{(n+1)}{(n-1)(n-2)(n-3)}\frac{(n-(k+1))(n-(k+2))}{(k+1)(k+2)(k+3)(k+4)}.
\end{array}
$$
We now put this together as

$$
\begin{array}{rcl}
\variance{\EcYn{1_{k}^{(n)}}} & = &
\binom{n}{2}^{-1}\pi_{n,k}
+ 2(n-2)\binom{n}{2}^{-1}\Expectation{1_{k,1}^{(n)}1_{k,2}^{(n)} \vert \mathrm{(ii)}}
\\&&
+\binom{n-2}{2}\binom{n}{2}^{-1}\Expectation{1_{k,1}^{(n)}1_{k,2}^{(n)} \vert \mathrm{(iii)}}
 - \pi_{n,k}^{2}
\end{array}
$$
and we obtain (through Mathematica)

$$
\begin{array}{rcl}
\variance{\EcYn{1_{k}^{(n)}}} & = &
4\frac{n+1}{n(n-1)^{2}}\frac{(n-(k+1))(n(3k^{2}+5k-4)-(k^{2}-k-8))}{(k+1)^{2}(k+2)^{2}(k+3)(k+4)}
\\& \to &
4\frac{3k^{2}+5k-4}{(k+1)^{2}(k+2)^{2}(k+3)(k+4)}.
\end{array}
$$
\end{proof}

\begin{lemma} \label{lemCov1k11k2}
For $k_{1} < k_{2}$ 
\be
\begin{array}{rcl}
\cov{\EcYn{1_{k_{2}}^{(n)}}}{\EcYn{1_{k_{1}}^{(n)}}} & = &
 \frac{(-8)(n+1)}{n(n-1)^{2}}\frac{(3n-(k_{2}-2))(n-(k_{2}+1))}{(k_{1}+1)(k_{1}+2)(k_{2}+1)(k_{2}+2)(k_{2}+3)(k_{2}+4)}.
\end{array}
\ee
\end{lemma}
\begin{proof}
Obviously

$$
\begin{array}{rcl}
\cov{\EcYn{1_{k_{1}}^{(n)}}}{\EcYn{1_{k_{2}}^{(n)}}} & = &
\Expectation{\EcYn{1_{k_{1}}^{(n)}}\EcYn{1_{k_{2}}^{(n)}}} - \Expectation{1_{k_{1}}}\Expectation{1_{k_{2}}}.
\end{array}
$$
We notice that

$$
\Expectation{\EcYn{1_{k_{1}}^{(n)}}\EcYn{1_{k_{2}}^{(n)}}}
= \Expectation{1_{k_{1}}^{(n)}1_{k_{2}}^{(n)}},
$$
where $1_{k_{1}}^{(n)}$, $1_{k_{2}}^{(n)}$ are the indicator variables
if two independently sampled pairs coalesced at speciation events
$k_{1}<k_{2}$ respectively. There are now two possibilities 
represented in Fig. \ref{fig2indeppairk1k2} (notice that
since $k_{1}\neq k_{2}$ the counterpart of event (i) 
in Fig. \ref{fig2pairphyl} cannot take place).
Event (ii) occurs with probability 
$4/(n+1)$ and (iii) with probability $(n-3)/(n+1)$.
Event (iii) can be divided into three ``subevents''. 

\begin{figure}[!htp]
\begin{center}
\includegraphics[width=0.9\textwidth]{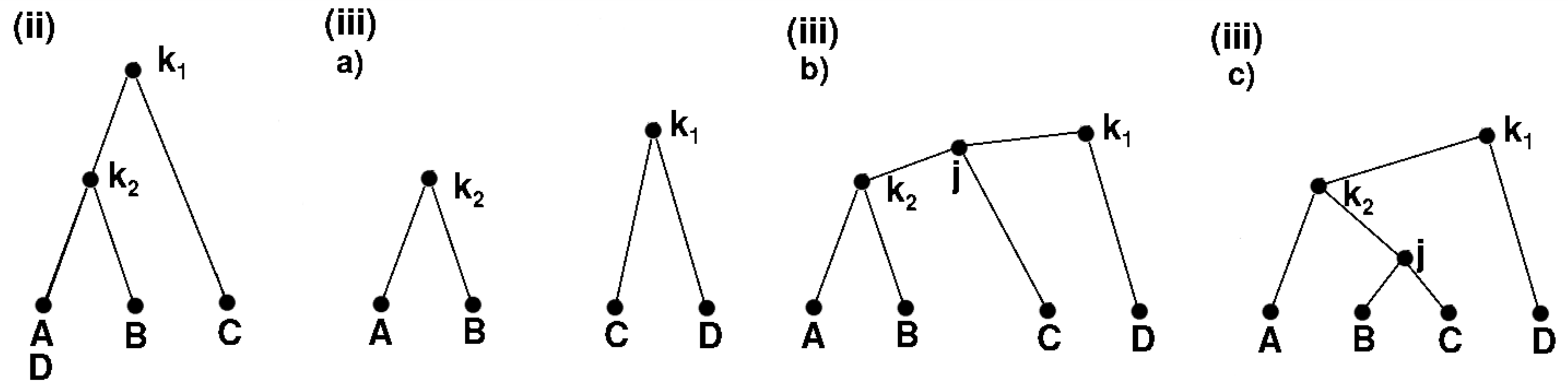}
\caption{The possible cases when drawing two random pairs of tip species that
coalesce at speciation events $k_{1}<k_{2}$ respectively. In the picture 
we ``randomly draw'' pairs $(A,B)$ and $(C,D)$.
\label{fig2indeppairk1k2}}
\end{center}
\end{figure}
Again we recall the proof of \citet{KBarSSag2015}'s Lemma $1$ and we write informally
for (ii) using Mathematica 

$$
\begin{array}{rcl}
\Expectation{1_{k_{1}}^{(n)}1_{k_{2}}^{(n)} \vert \mathrm{(ii)}}  & =&
\left(1-\frac{3}{\binom{n}{2}}\right)\ldots\left(1-\frac{3}{\binom{k_{2}+2}{2}}\right)\frac{1}{\binom{k_{2}+1}{2}}
\left(1-\frac{1}{\binom{k_{2}}{2}}\right)\ldots
\\ && \cdot
\left(1-\frac{1}{\binom{k_{1}+2}{2}}\right)\frac{1}{\binom{k_{1}+1}{2}}
\\&=&
4\frac{(n+1)(n+2)}{(n-1)(n-2)}\frac{1}{(k_{1}+1)(k_{1}+2)(k_{2}+2)(k_{2}+3)}.
\end{array}
$$
In the same way for the subcases of (iii) 

$$
\begin{array}{rcl}
\Expectation{1_{k_{1}}^{(n)}1_{k_{2}}^{(n)} \vert \mathrm{(iii)}}  & =&
\left(1-\frac{6}{\binom{n}{2}}\right)\ldots\left(1-\frac{6}{\binom{k_{2}+2}{2}}\right)\frac{1}{\binom{k_{2}+1}{2}}
\\&&\cdot
\left(1-\frac{3}{\binom{k_{2}}{2}}\right)\ldots\left(1-\frac{3}{\binom{k_{1}+2}{2}}\right)\frac{1}{\binom{k_{1}+1}{2}}
\\&&+
\sum\limits_{j=k_{1}+1}^{k_{2}-1}
\left(1-\frac{6}{\binom{n}{2}}\right)\ldots\left(1-\frac{6}{\binom{k_{2}+2}{2}}\right)\frac{1}{\binom{k_{2}+1}{2}}
\\&&\cdot
\left(1-\frac{3}{\binom{k_{2}}{2}}\right)\ldots\left(1-\frac{3}{\binom{j+2}{2}}\right)\frac{2}{\binom{j+1}{2}}
\left(1-\frac{1}{\binom{j}{2}}\right)\ldots
\\ && \cdot
\left(1-\frac{1}{\binom{k_{1}+2}{2}}\right)\frac{1}{\binom{k_{1}+1}{2}}
\\&&+
\sum\limits_{j=k_{2}+1}^{n-1}
\left(1-\frac{6}{\binom{n}{2}}\right)\ldots\left(1-\frac{6}{\binom{j+2}{2}}\right)\frac{4}{\binom{j+1}{2}}
\\ &&\cdot
\left(1-\frac{3}{\binom{j}{2}}\right)\ldots\left(1-\frac{3}{\binom{k_{2}+2}{2}}\right)\frac{2}{\binom{k_{2}+1}{2}}
\left(1-\frac{1}{\binom{k_{2}}{2}}\right)\ldots
\\ && \cdot
\left(1-\frac{1}{\binom{k_{1}+2}{2}}\right)\frac{1}{\binom{k_{1}+1}{2}}
\\&=&
4\frac{(n+2)(n+1)}{(n-1)(n-2)(n-3)}\frac{n(k_{2}+6)-5k_{2}-14}{(k_{1}+1)(k_{1}+2)(k_{2}+2)(k_{2}+3)(k_{2}+4)}.
\end{array}
$$
We now put this together as 

$$
\begin{array}{rcl}
\cov{\EcYn{1_{k_{1}}^{(n)}}}{\EcYn{1_{k_{2}}^{(n)}}}
& = &
 2(n-2)\binom{n}{2}^{-1}\Expectation{1_{k_{1}}^{(n)}1_{k_{2}}^{(n)} \vert \mathrm{(ii)}}
\\&& 
+ \binom{n-2}{2}\binom{n}{2}^{-1}\Expectation{1_{k_{1}}^{(n)}1_{k_{2}}^{(n)} \vert \mathrm{(iii)}}
- \pi_{n,k_{1}}\pi_{n,k_{2}}
\end{array}
$$
and we obtain

$$
\begin{array}{rcl}
\cov{\EcYn{1_{k_{1}}^{(n)}}}{\EcYn{1_{k_{2}}^{(n)}}}
& = &
 \frac{(-8)(n+1)}{n(n-1)^{2}}\frac{(3n-(k_{2}-2))(n-(k_{2}+1))}{(k_{1}+1)(k_{1}+2)(k_{2}+1)(k_{2}+2)(k_{2}+3)(k_{2}+4)}
\\ & \to &
(-24)\frac{1}{(k_{1}+1)(k_{1}+2)(k_{2}+1)(k_{2}+2)(k_{2}+3)(k_{2}+4)}.
\end{array}
$$

\end{proof}

\begin{thm}\label{thmEVin}

\be
\Expectation{V_{i}^{(n)}}=2\frac{1}{n-1} \frac{n-i}{i(i+1)}
\ee
\end{thm}
\begin{proof}
We immediately have

$$
\begin{array}{rcl}
\Expectation{V_{i}^{(n)}} & = & \frac{1}{i}\sum\limits_{k=i}^{n-1}\Expectation{\EcYn{1_{k}^{(n)}}} 
\\ & = &
2\frac{n+1}{n-1}\frac{1}{i}\sum\limits_{k=i}^{n-1} \frac{1}{(k+1)(k+2)}
\\ & = & 
2\frac{1}{n-1} \frac{n-i}{i(i+1)} 
\\ & \to & \frac{2}{i(i+1)}.
 \end{array}
$$
\end{proof}

\begin{thm}\label{thmVarVin}

\be
\variance{V_{i}^{(n)}}=4\frac{(n+1)}{n(n-1)^{2}}\frac{(n-i)(n-(i+1))(i-1)}{i^{2}(i+1)^{2}(i+2)(i+3)}
\ee
\end{thm}
\begin{proof}
We immediately may write using Lemmata \ref{lemVar1k}, \ref{lemCov1k11k2} and Mathematica

$$
\begin{array}{rcl}
\variance{V_{i}^{(n)}} & = & \frac{1}{i^{2}}\left(\sum\limits_{k=i}^{n-1}\variance{\EcYn{1_{k}^{(n)}}} 
+ 2\sum\limits_{i=k_{1}<k_{2}}^{n-1}\cov{\EcYn{1_{k_{1}}^{(n)}}}{\EcYn{1_{k_{2}}^{(n)}}} \right)
\\ & = &
\frac{4}{i^{2}}\left(\sum\limits_{k=i}^{n-1}
\frac{n+1}{n(n-1)^{2}}\frac{(n-(k+1))(n(3k^{2}+5k-4)-(k^{2}-k-8))}{(k+1)^{2}(k+2)^{2}(k+3)(k+4)}
\right. \\ && \left. -4 \sum\limits_{i=k_{1}<k_{2}}^{n-1}
\frac{(n+1)}{n(n-1)^{2}}\frac{(3n-(k_{2}-2))(n-(k_{2}+1))}{(k_{1}+1)(k_{1}+2)(k_{2}+1)(k_{2}+2)(k_{2}+3)(k_{2}+4)}\right)
\\ & = &
4\frac{(n+1)}{n(n-1)^{2}}\frac{(n-i)(n-(i+1)(i-1)}{i^{2}(i+1)^{2}(i+2)(i+3)}
\\ & \to &
4\frac{(i-1)}{i^{2}(i+1)^{2}(i+2)(i+3)}.
\end{array}
$$
\end{proof}

\begin{thm}
For $1\le i_{1} < i_{2} \le n-1$ we have

\be
\cov{V_{i_{1}}^{(n)}}{V_{i_{2}}^{(n)}}=
4\frac{(n+1)}{n(n-1)^{2}}
\frac{(i_{1}-1)(n-i_{2})(n-(i_{2}+1))}{i_{1}(i_{1}+1)i_{2}(i_{2}+1)(i_{2}+2)(i_{2}+3)}.
\ee
\end{thm}
\begin{proof}
Again using Lemmata \ref{lemVar1k}, \ref{lemCov1k11k2}, Mathematica and the fact that $i_{1} < i_{2}$

$$
\begin{array}{l}
\cov{V_{i_{1}}^{(n)}}{V_{i_{2}}^{(n)}} =
\frac{1}{i_{1}i_{2}}\left(\cov{\sum\limits_{k=i_{1}}^{n-1}\EcYn{1_{k}^{(n)}}}{\sum\limits_{k=i_{2}}^{n-1}\EcYn{1_{k}^{(n)}}} \right)
\\  = 
\frac{1}{i_{1}i_{2}}\left(
\variance{\sum\limits_{k=i_{2}}^{n-1}\EcYn{1_{k}^{(n)}}}
+ \cov{\sum\limits_{k=i_{1}}^{i_{2}-1}\EcYn{1_{k}^{(n)}}}{\sum\limits_{k=i_{2}}^{n-1}\EcYn{1_{k}^{(n)}}} \right)
\\  =
\frac{1}{i_{1}i_{2}}\left(
(i_{2}^{2})\variance{V_{i_{2}}^{(n)}}
+ \sum\limits_{k_{1}=i_{1}}^{i_{2}-1}\sum\limits_{k_{2}=i_{2}}^{n-1}
\cov{\EcYn{1_{k_{1}}^{(n)}}}{\sum\limits_{k=i_{2}}^{n-1}\EcYn{1_{k_{2}}^{(n)}}} \right)
\\  = 
4\frac{(n+1)}{n(n-1)^{2}}
\frac{(i_{1}-1)(n-i_{2})(n-(i_{2}+1)}{i_{1}(i_{1}+1)i_{2}(i_{2}+1)(i_{2}+2)(i_{2}+3)}
\\  \to 
4\frac{i_{1}-1}{i_{1}(i_{1}+1)i_{2}(i_{2}+1)(i_{2}+2)(i_{2}+3)}.
\end{array}
$$
\end{proof}

\begin{thm}\label{thmVarWn}

\be
\begin{array}{rcl}
\variance{\sum\limits_{i=1}^{n-1}V_{i}^{(n)}} &=& 
\frac{1}{54n^{2}(n-1)^{2}}\left(
179n^{4}+588n^{3}+133n^{2}-432n
\right. \\&& \left.
-468-108n^{2}(n+1)(n+3)H_{n-1,2}
\right. \\&& \left.
-144nH_{n-1,1}\right)
\to  \frac{179}{54}-\frac{\pi^{2}}{3} \approx 1.347,\\
\variance{\sum\limits_{i=1}^{n-1}V_{i}^{(n)}Z_{i}} &=& 
\frac{1}{9n^{2}(n-1)^{2}}
\left(
12n^{2}(n^{2}-6n-4)H_{n-1,2} -9n^{4}
\right. \\ && \left.
+102n^{3}
+51n^{2}-24nH_{n-1,1}-72n-72
\right)
\\ & \to & \frac{2}{9}\pi^{2}-1 \approx 1.193, \\
\variance{\sum\limits_{i=1}^{n-1}\Expectation{V_{i}^{(n)}}Z_{i}} &=& 
\frac{2}{3n^{2}(n-1)^{2}}\left(\left(12H_{n-1,2}-18\right)n^{4}-24n^{3}
\right. \\ && \left.
+12n^{2}(2n+1)H_{n-1,2}-24n^{2}+24n+12
\right)
\\ & \to & \frac{4}{3}\pi^{2}-12 \approx 1.159, \\
\variance{\sum\limits_{i=1}^{n-1}\left(V_{i}^{(n)}-\Expectation{V_{i}^{(n)}}\right)Z_{i}} &=&
\frac{1}{9n^{2}(n-1)^{2}}\left(
99n^{4}+174n^{3}-21n^{2}-144n
\right. \\ && \left.
-108
-12n^{2}(n+1)(5n+7)H_{n-1,2}
\right. \\ && \left.
-24nH_{n-1,1} \right)
 \to  11-\frac{10}{9}\pi^{2} \approx 0.034.
\end{array}
\ee
\end{thm}
\begin{proof}
We use Mathematica to first calculate 

$$
\begin{array}{rcl}
\variance{\sum\limits_{i=1}^{n-1}V_{i}^{(n)}} & = & 
\sum\limits_{i=1}^{n-1}\variance{V_{i}^{(n)}}
+ 2\sum\limits_{1=i_{1}<i_{2}}^{n-1}\cov{V_{i_{1}}^{(n)}}{V_{i_{2}}^{(n)}} 
\\ & = &
\frac{1}{54n^{2}(n-1)^{2}}\left(
179n^{4}-108n^{2}(n+1)(n+3)H_{n-1,2}+588n^{3}
\right. \\ && \left.
+133n^{2}-144nH_{n-1,1}-432n-468
\right)
\\ & \to &
\frac{179}{54}-\frac{\pi^{2}}{3} \approx 1.347 .
\end{array}
$$
For the second we again use Mathematica and the fact that the 
$Z_{i}$s are i.i.d. $\exp(1)$.

$$
\begin{array}{rcl}
\variance{\sum\limits_{i=1}^{n-1}\Expectation{V_{i}^{(n)}}Z_{i}} & = &
\sum\limits_{i=1}^{n-1}\left(2\frac{1}{n-1}\frac{n-i}{i(i+1)}\right)^{2}
\\ & = &
\frac{2\left(12H_{n-1,2}-18\right)n^{4}+2\left(6n^{2}(2n+1)H_{n-1,2}-12n^{3}-12n^{2}+12n+6 \right)}{3n^{2}(n-1)^{2}}
\\ & \to &
\frac{4}{3}\pi^{2}-12 \approx 1.159.
\end{array}
$$
For the third equality we use Mathematica and the fact that for independent families
$\{X\}$ and $\{Y\}$ of 
random variables we have 

$$
\begin{array}{rcl}
\variance{XY} & = & \Expectation{Y^{2}}\variance{X} + (\Expectation{X})^{2}\variance{Y}, \\
\cov{X_{1}Y_{1}}{X_{2}Y_{2}} & = & \Expectation{Y_{1}}\Expectation{Y_{2}}\cov{X_{1}}{X_{2}} 
+ \Expectation{X_{1}}\Expectation{X_{2}}\cov{Y_{1}}{Y_{2}}.
\end{array}
$$
As the $Z_{i}$s are i.i.d. $\exp(1)$ we use Mathematica to obtain

$$
\begin{array}{l}
\variance{\sum\limits_{i=1}^{n-1}V_{i}^{(n)}Z_{i}}  = 
\sum\limits_{i=1}^{n-1}\variance{V_{i}^{(n)}Z_{i}}
+ 2\sum\limits_{1=i_{1}<i_{2}}^{n-1}\cov{V_{i_{1}}^{(n)}Z_{i_{1}}}{V_{i_{2}}^{(n)}Z_{i_{2}}} 
\\  = 
2\sum\limits_{i=1}^{n-1}\variance{V_{i}^{(n)}}+
\sum\limits_{i=1}^{n-1}\left(\Expectation{V_{i}^{(n)}}\right)^{2}+
 2\sum\limits_{1=i_{1}<i_{2}}^{n-1}\cov{V_{i_{1}}^{(n)}}{V_{i_{2}}^{(n)}} 
\\  = 
\frac{1}{9n^{2}(n-1)^{2}}
\left(
12n^{2}(n^{2}-6n-4)H_{n-1,2} -9n^{4}+102n^{3}
\right. \\  \left.
+51n^{2}-24nH_{n-1,1}-72n-72
\right)
\\  \to 
\frac{1}{9}\left(2\pi^{2}-9\right) \approx 1.193.
\end{array}
$$
For the fourth equality we use the same properties and pair--wise independence of $Z_{i}$s.

$$
\begin{array}{l}
\variance{\sum\limits_{i=1}^{n-1}\left(V_{i}^{(n)}-\Expectation{V_{i}^{(n)}}\right)Z_{i}}  = 
\sum\limits_{i=1}^{n-1}\variance{V_{i}^{(n)}Z_{i}}
+\sum\limits_{i=1}^{n-1}\variance{\Expectation{V_{i}^{(n)}}Z_{i}}
\\  - 2\sum\limits_{1=i_{1}<i_{2}}^{n-1}\cov{V_{i_{1}}^{(n)}Z_{i_{1}}}
{\Expectation{V_{i_{2}}^{(n)}}Z_{i_{2}}} 
\\  = 
\sum\limits_{i=1}^{n-1}\variance{V_{i}^{(n)}Z_{i}}
+\sum\limits_{i=1}^{n-1}\left(\Expectation{V_{i}^{(n)}}\right)^{2}
 - 2\sum\limits_{i=1}^{n-1}\left(\Expectation{V_{i}^{(n)}}\right)^{2}
\\  = 
\sum\limits_{i=1}^{n-1}\variance{V_{i}^{(n)}Z_{i}}
-\sum\limits_{i=1}^{n-1}\left(\Expectation{V_{i}^{(n)}}\right)^{2}
\\  = 
\frac{1}{9n^{2}(n-1)^{2}}\left(
99n^{4}-12n^{2}(n+1)(5n+7)H_{n-1,2}+174n^{3}-21n^{2}
\right. \\  \left.
-24nH_{n-1,1}-144n-108
\right)
 \to 
11 - \frac{10}{9}\pi^{2} \approx 0.034.
\end{array}
$$

\end{proof}
It is worth noting that the above Lemmata and Theorems were confirmed by numerical evaluations
of the formulae and comparing these to simulations performed to obtain Fig. \ref{figSimulASlimPhi}.
As a check also notice that, as implied by variance properties,

$$
\begin{array}{l}
\variance{\sum\limits_{i=1}^{n-1}\Expectation{V_{i}^{(n)}}Z_{i}}
+\variance{\sum\limits_{i=1}^{n-1}\left(V_{i}^{(n)}-\Expectation{V_{i}^{(n)}}\right)Z_{i}}
\\
\to  \frac{4}{3}\pi^{2}-12+11-\frac{10}{9}\pi^{2} 
 = \frac{2}{9}\pi^{2}-1 \leftarrow
\variance{\sum\limits_{i=1}^{n-1}V_{i}^{(n)}Z_{i}}.
\end{array}
$$

\begin{thm}\label{thmExpVarCondCLTv2}

\be\label{eqECondVar}
\begin{array}{l}
\Expectation{\variance{\left(n\right)^{-1}\sum\limits_{i=2}^{n-2}\frac{V_{i}^{(n)}Z_{i} - \Expectation{V_{i}^{(n)}}}{\sqrt{\variance{V_{i}^{(n)}}}} \Bigg\vert \{V_{i}^{(n)} \}}} \to 0.5.
\end{array}
\ee
\end{thm}

\begin{proof}
Using the limit for the variance of $V_{i}^{(n)}$ (Thm. \ref{thmVarVin}) and the independence of the $Z_{i}$s we have

$$
\begin{array}{l}
\Expectation{\variance{\frac{1}{n}\sum\limits_{i=2}^{n-2}\frac{V_{i}^{(n)}Z_{i} - \Expectation{V_{i}^{(n)}}}{\sqrt{\variance{V_{i}^{(n)}}}} \Bigg\vert \{V_{i}^{(n)} \}}} 
\sim
\frac{1}{4n^{2}}\sum\limits_{i=2}^{n-2}\frac{i^{2}(i+1)^{2}(i+2)(i+3)}{(i-1)}\Expectation{\right(V_{i}^{(n)}\left)^{2}}.
\end{array}
$$
Now from Thms. \ref{thmVarVin} and \ref{thmEVin} we have

$$
\begin{array}{l}
\Expectation{\right(V_{i}^{(n)}\left)^{2}} = 4\frac{n+1}{n(n-1)^{2}}\frac{(n-i)(n-(i+1))(i-1)}{i^{2}(i+1)^{2}(i+2)(i+3)}  +
\left(\frac{2}{n-1}\frac{n-i}{i(i+1)} \right)^{2}
\\=4\frac{1}{(n-1)^{2}}\frac{(n-i)^{2}}{i^{2}(i+1)^{2}}\left(\frac{n+1}{n(n-i)}\frac{(n-(i+1))(i-1)}{(i+2)(i+3)} +1\right)
\to  4\frac{i+5}{i^{2}(i+1)(i+2)(i+3)}. 
\end{array}
$$
Plugging this in (and using Mathematica)

$$
\begin{array}{l}
\frac{1}{4n^{2}}\sum\limits_{i=2}^{n-2}\frac{i^{2}(i+1)^{2}(i+2)(i+3)}{(i-1)}\Expectation{\right(V_{i}^{(n)}\left)^{2}}
\sim
\frac{1}{4n^{2}}\sum\limits_{i=2}^{n-2}\frac{i^{2}(i+1)^{2}(i+2)(i+3)4(i+5)}{(i-1)i^{2}(i+1)(i+2)(i+3)}
\\=
n^{-2}\sum\limits_{i=2}^{n-2}\frac{(i+1)(i+5)}{(i-1)}
=n^{-2}\frac{1}{2}\left(n^{2}+11n+24H_{n,1}-42 \right) \to 0.5.
\end{array}
$$

\end{proof}

\begin{remark}\label{remCLTconj}
Simulations
presented in Fig. \ref{figSimulNCLT} and Thm. \ref{thmExpVarCondCLTv2} suggest a
different possible CLT, namely

\be\label{eqCLTn}
\left(n\right)^{-1}\sum\limits_{i=2}^{n-2}\frac{V_{i}^{(n)}Z_{i} - \Expectation{V_{i}^{(n)}}}{\sqrt{\variance{V_{i}^{(n)}}}} 
\stackrel{\mathrm{weakly}}{\longrightarrow} \mathrm{some~distribution}(\mathrm{mean}=0,\mathrm{variance}=\frac{1}{2}).
\ee
We sum over $i=2,\ldots n-2$ as $V_{1}^{(n)}=1$ and $V_{n-1}^{(n)}=\binom{n}{2}^{-1}$ for all $n$.
It would be tempting to take the distribution to be a normal one. 
However, we should be wary after Rem. \ref{remWbar} and Fig. \ref{figSimulASlimPhi} that for our rather
delicate problem even very fine simulations can indicate incorrect weak limits. It remains
to study the variance of the conditional variance in Eq. \eqref{eqECondVar}. It is not
entirely clear if this variance of the conditional variance will converge to $0$. Hence,
it remains an open problem to investigate the conjecture of Eq. \eqref{eqCLTn}.
\end{remark}

\begin{figure}[!htp]
\centering
\includegraphics[width=0.47\textwidth]{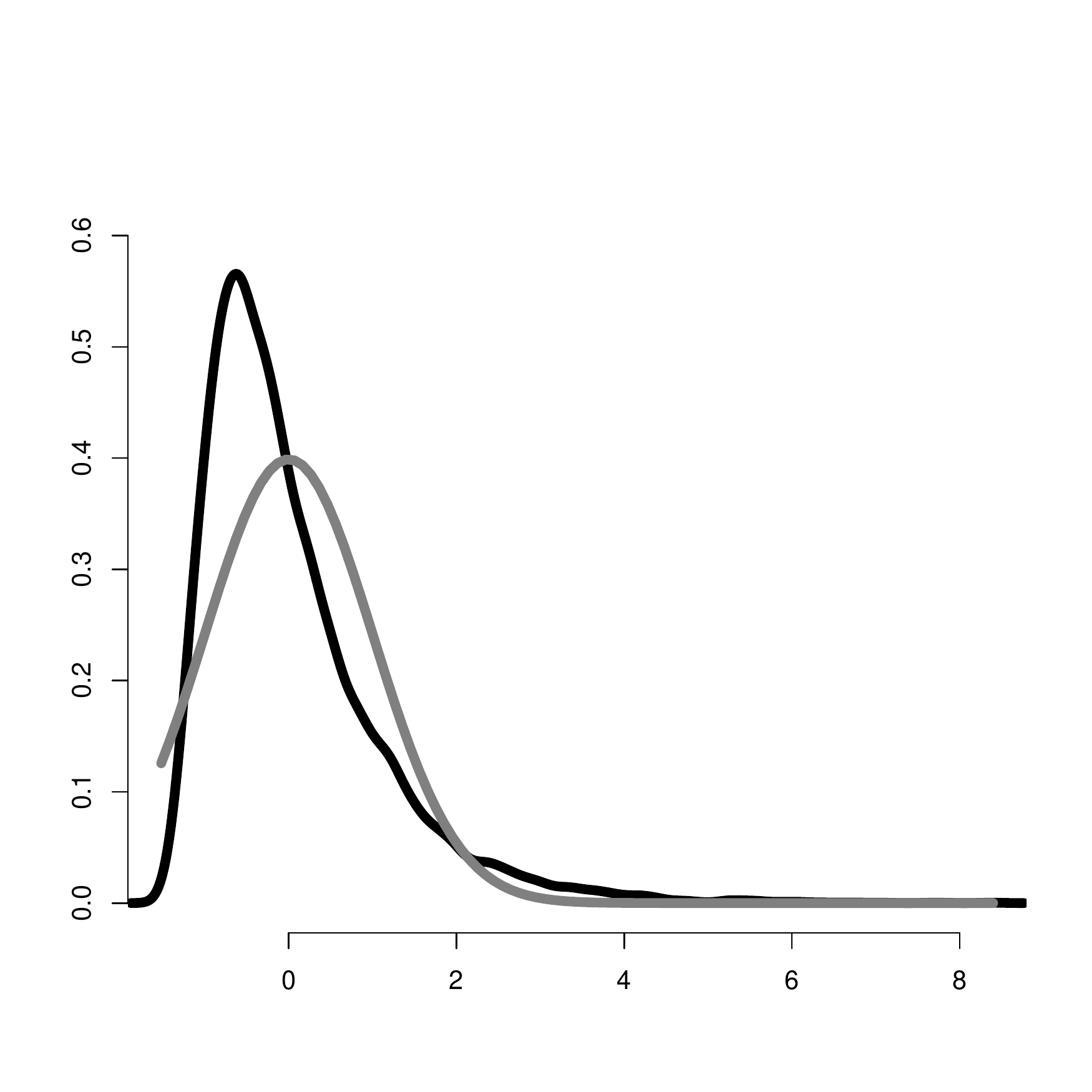}
\includegraphics[width=0.47\textwidth]{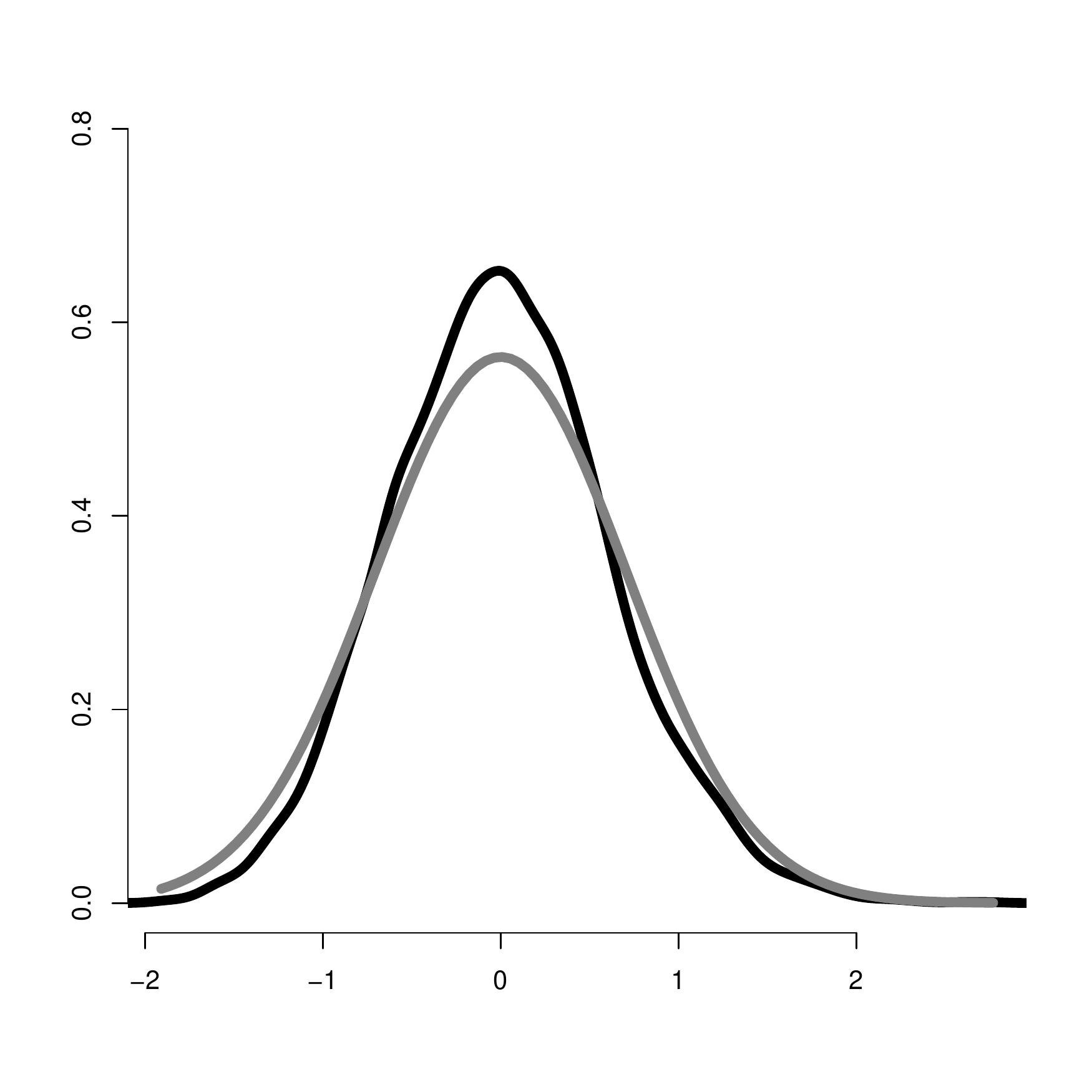}
\caption{
Density estimates of scaled and centred cophenetic indices for $10000$ simulated $500$ tip Yule trees with $\lambda=1$. 
Left: density estimate of $(\Phi^{(n)}-\E{\Phi^{(n)}})/\sqrt{\var{\Phi^{(n)}}}$.
The black curve is the density fitted to simulated data by R's \texttt{density()} function, the gray 
is the $\mathcal{N}(0,1)$ density.
Right: simulation of Eq. \eqref{eqCLTn}, the gray curve is the $\mathcal{N}(0,1/2)$ density,
and the black curve is the density fitted to simulated data by R's \texttt{density()} function.
The sample variance of the simulated Eq. \eqref{eqCLTn} values is $0.385$ indicating
that with $n=500$ we still have a high variability or alternatively that the variance
of the sample variance in Eq. \eqref{eqECondVar} does not converge to $0$.
\label{figSimulNCLT}
}
\end{figure}

\section{Alternative descriptions}\label{secAltDesc}
\subsection{Difference process}\label{secDiffProc}
Let us consider in detail the families of random variables $V_{i}^{(n)}$ and $\EcYn{1_{k}^{(n)}}$.
Obviously $V_{i}^{(n)}$ is $\binom{n}{2}i$ times the number of pairs that coalesced
after the $i-1$ speciation event for a given Yule tree. Denote 

$$A_{i}^{(n)}:=iV_{i}^{(n)}.$$
As going from $n$ to $n+1$ means a new speciation event and coalescent at this new $n$th event,
then 

$$A_{i}^{(n+1)} \ge \binom{n+1}{2}^{-1}\left(\binom{n}{2}A_{i}^{(n)}+1 \right).$$
We also know by previous calculations that 

$$\Expectation{A_{i}^{(n)}}=i\Expectation{V_{i}^{(n)}}=2(n-i)/((n-1)(i+1)) \to 2/(i+1).$$
Let $\binom{n+1}{2}\epsilon^{(n)}_{i}$ denote the number of newly introduced coalescent
events after the $(i-1)$--one when we go from $n$ to $n+1$ species. Obviously 
$\epsilon^{(n)}_{i}>\binom{n+1}{2}$. Then, we may write

$$
A_{i}^{(n+1)} = \binom{n+1}{2}^{-1}\binom{n}{2}A_{i}^{(n)}+\epsilon^{(n)}_{i}.
$$
Now, 

$$
\begin{array}{l}
\Expectation{\epsilon^{(n)}_{i}} = 
\Expectation{A_{i}^{(n+1)}} - \binom{n+1}{2}^{-1}\binom{n}{2}\Expectation{A_{i}^{(n)}}
= \frac{2(n+1-i)}{n(i+1)} - \frac{n(n-1)}{n(n+1)}\frac{2(n-i)}{(n-1)(i+1)}
\\
= \frac{2}{i+1}\left(\frac{n+1-i}{n} - \frac{n-i}{n+1}\right)
= \frac{2}{i+1}\frac{(n-i+1)(n+1)-n(n-i)}{n(n+1)}
= \frac{2}{i+1}\frac{n(n-i)+n+n-i+1-n(n-i)}{n(n+1)}
\\ = \frac{2}{i+1}\frac{2n+1-i}{n(n+1)} \to 0.
\end{array}
$$
Therefore, for every $i$, $\epsilon^{(n)}_{i}\to 0$ almost surely as 
it is a positive random variable whose expectation goes to $0$.
However, $A^{(n)}_{i}$ is bounded by $1$, as it can be understood
in terms of the conditional (on tree) cumulative distribution function for the random variable
$\kappa$---at which speciation event did a random pair of tips coalesce, i.e.
for all $i=1,\ldots,n-1$

$$
P(\kappa \le i-1 \vert \mathcal{Y}_{n}) = 1-A^{(n)}_{i}.
$$
Therefore, as $A^{(n)}_{i}$ is bounded by $1$ and the difference process

$$A^{(n)}_{i}-A^{(n-1)}_{i}=\epsilon^{(n)}_{i}$$ 
goes almost surely to $0$ we
may conclude that $A^{(n)}_{i}$ converges almost surely to some random variable
$A_{i}$. 
In particular, this implies the almost sure convergence of $V^{(n)}_{i}$ to a limiting
random variable $V_{i}$.
Furthermore, as $\Expectation{\sum_{i=1}^{n-1}V^{(n)}_{i}}$ and
$\variance{\sum_{i=1}^{n-1}V^{(n)}_{i}}$ are both $O(1)$ we may conclude that 
$\sum_{i=1}^{n-1}V^{(n)}_{i}$ also converges almost surely. This means that the
discrete version (all $T_{i}=1$, corresponding to $\tilde{\Phi}^{(n)}$) 
of the cophenetic index converges almost surely (compare with Thm. \ref{thmWnDiscConv}).

\subsection{Poly\'a urn description}\label{secPolyaDesc}
The cophenetic index both in the discrete and continuous version has the following Poly\'a urn description.
We start with an urn filled with $n$ balls. Each ball has a number painted on it, $0$ initially. 
At each step we remove a pair of balls, say with numbers $x$ and $y$ and return a ball
with the number $(x+1)(y+1)$ painted on it. We stop when there is only one ball, it will have value $\binom{n}{2}$.
Denote $B_{k,i,n}$ as the value painted on the $k$--th ball in the $i$--th step when we initially 
started with $n$ balls. 
Then we can represent the cophenetic index as 

$$\Phi^{(n)}=\sum\limits_{i=1}^{n-1} \left(\sum\limits_{k=1}^{i}B_{k,i,n}\right) T_{i} 
~~\mathrm{and}~~
\tilde{\Phi}^{(n)}=\sum\limits_{i=1}^{n-1} \sum\limits_{k=1}^{i}B_{k,i,n}.$$

\section*{Acknowledgments}
I was supported by the Knut and Alice Wallenberg Foundation
and am now by the Swedish Research Council (Vetenskapsr\aa det) grant no. $2017$--$04951$.
I am grateful to the Barcelona Graduate School of Mathematics (BGSMath)
for sponsoring the Workshop on Algebraical and Combinatorial Phylogenetics
which significantly contributed to the development of my work.
I would like to thank the whole 
Computational Biology and Bioinformatics Research Group of the Balearic Islands University
for hosting me on multiple occasions, many discussions and suggestions on phylogenetic indices.
My visits to the Balearic Islands University were partially supported by the 
the G S Magnuson Foundation of the Royal Swedish Academy of Sciences
(grants no. MG$2015$--$0055$, MG$2017$--$0066$)
and The Foundation for Scientific Research and Education in Mathematics (SVeFUM).
I would like to acknowledge Gabriel Yedid for numerous discussions on the 
distribution of the cophenetic index and sharing his cophenetic index simulation R code. 
I am grateful to 
Cecilia Holmgren and Svante Janson for pointing me to the 
works on contraction--type distributions and many discussions.
I would furthermore like to acknowledge 
Wojciech Bartoszek, Sergey Bobkov, Joachim Domsta, Serik Sagitov, Mike Steel
for helpful comments and discussions related to this work.
I am indebted to two anonymous reviewers, an anonymous editor and Haochi Kiang for careful reading of an earlier version
of the manuscript and comments significantly improving it.

\bibliographystyle{plainnat}
\bibliography{PhyloTrees.bib}
\clearpage

\section*{Appendix A: Mathematica code for Section \ref{sec2ndOrd}}
\begin{lstlisting}
(* 
Mathematica code used to obtain the closed form formulae 
of Section 3. Second order properties in K. Bartoszek 
Exact and approximate limit behaviour of the Yule tree's 
cophenetic index.
The script was run using Mathematica 9.0 for Linux x86 
(64-bit) running on Ubuntu 12.04.5 LTS. It has to be noted 
that Mathematica's output should be manually postprocessed 
in order to have the formulae in terms of harmonic sums and 
not derivatives of polygamma functions. 
All the references in this script point to appropriate fragments 
of the manuscript. 
*)

(* 
We choose the pairs in order, i.e. first the first pair to 
coalesce then the second pair to coalesce. 
*)
(* Compare with proof of Lemma 1 of Bartoszek and Sagitov (2015b) *)
FcoalProb[n_,k_,c_]=FullSimplify[Product[(1-c/((r*(r-1))/2))
,{r,k+2,n}]]

(* Def. 2.3, Eq. (1) *)
E1k[n_,k_]:=(2*(n+1)/((n-1)*(k+1)*(k+2)))
(* Lemma 6.1, Eq. (18) *)
Var1k[n_,k_]:=(E1k[n,k]-E1k[n,k]*E1k[n,k])
(* Lemma 6.2, Eq. (19) *)
Cov1k11k2[n_,k1_,k2_]:=(-E1k[n,k1]*E1k[n,k2])

(* Lemma 6.3, Eq. (20) *)
VarE1k[n_,k_]=(FullSimplify[
(1/(n*(n-1)/2))*(2*(n+1)/((n-1)*(k+1)*(k+2)))
+
(2*(n-2)/(n*(n-1)/2))*(Sum[FcoalProb[n,j,3]*(1/((j+1)*j/2))
*FcoalProb[j,k,1]*(1/((k+1)*k/2)),{j,k+1,n-1}])
+
((n-2)*(n-3)/2/(n*(n-1)/2))*(
Sum[Sum[FcoalProb[n,j1,6]*(4/((j1+1)*j1/2))
*FcoalProb[j1,j2,3]*(1/((j2+1)*j2/2))
*FcoalProb[j2,k,1]*(1/((k+1)*k/2)),{j1,j2+1,n-1}],{j2,k+1,n-2}]
)- (E1k[n,k])*(E1k[n,k])])

(* Lemma 6.4, Eq. (21) *)
CovE1k1E1k2[n_,k1_,k2_]=(FullSimplify[
(2*(n-2)/(n*(n-1)/2))*(FcoalProb[n,k2,3]*(1/((k2+1)*k2/2))
*FcoalProb[k2,k1,1]*(1/((k1+1)*k1/2)))
+
((n-2)*(n-3)/2/(n*(n-1)/2))*(FcoalProb[n,k2,6]
*(1/((k2+1)*k2/2))*FcoalProb[k2,k1,3]*(1/((k1+1)*k1/2))
+
Sum[FcoalProb[n,k2,6]*(1/((k2+1)*k2/2))*FcoalProb[k2,j,3]
*(2/((j+1)*j/2))*FcoalProb[j,k1,1]*(1/((k1+1)*k1/2)),{j,k1+1,k2-1}]
+
Sum[FcoalProb[n,j,6]*(4/((j+1)*j/2))*FcoalProb[j,k2,3]
*(1/((k2+1)*k2/2))*FcoalProb[k2,k1,1]*(1/((k1+1)*k1/2)),{j,k2+1,n-1}]
)-(E1k[n,k1])*(E1k[n,k2])])

(* Thm. 6.1, Eq. (22) *)
EVi[n_,i_]:=(FullSimplify[Sum[E1k[n,k],{k,i,n-1}]/i])

(* Thm. 6.2, Eq. (23) *)
VarVi[n_,i_]=(FullSimplify[(Sum[VarE1k[n,k],{k,i,n-1}]
+2*Sum[Sum[CovE1k1E1k2[n,k1,k2],{k2,k1+1,n-1}],{k1,i,n-1}])/(i*i)])

(* Thm. 6.3, Eq. (24) *)
CovVi1Vi2[n_,i1_,i2_]=(FullSimplify[(i2*i2*VarVi[n,i2]
+Sum[Sum[CovE1k1E1k2[n,k1,k2],{k2,i2,n-1}],{k1,i1,i2-1}])/(i1*i2)])

(* Thm. 6.4, formula 1 *)
EVi2[n_,i_]=(FullSimplify[VarVi[n,i]+(EVi[n,i]^2)])
VarSumVi[n_]=(FullSimplify[Sum[EVi2[n,i],{i,1,n-1}]
+2*Sum[Sum[CovVi1Vi2[n,i1,i2],{i2,i1+1,n-1}],{i1,1,n-2}]])

(* Thm. 6.4 Eq. (13), formula 2 *)
VarWn[n_]=(FullSimplify[2*Sum[VarVi[n,i],{i,1,n-1}]
+Sum[(EVi[n,i])^2,{i,1,n-1}]+2*Sum[Sum[CovVi1Vi2[n,i1,i2]
,{i2,i1+1,n-1}],{i1,1,n-1}]])

(* Thm. 6.4 Eq. (13), formula 3 *)
VarWnBar[n_]=(FullSimplify[Sum[(EVi[n,i])^2,{i,1,n-1}]])

(* Thm. 6.4 Eq. (13), formula 4 *)
VarWnCentre[n_]=(FullSimplify[2*Sum[VarVi[n,i],{i,1,n-1}]
+2*Sum[Sum[CovVi1Vi2[n,i1,i2],{i2,i1+1,n-1}],{i1,1,n-1}]])

(* Thm. 6.5 *)
FinalPart[n_]=(Sum[(i+1)*(i+5)/(i-1),{i,2,n-2}])

\end{lstlisting}

\section*{Appendix B: Counterparts of \citet{URos1991}'s Prop. $3.2$ for the cophenetic index}\label{appLemProofs}
\begin{lemma}\label{lemCCDiscdiff}
Define for $i \in \{1,\ldots, n\}$

$$
\tilde{C}_{n}(i) = n^{-2} \left(\Expectation{\tilde{\Phi}^{(i)}}+\Expectation{\tilde{\Phi}^{(n-i)}}-\Expectation{\tilde{\Phi}^{(n)}} + \binom{n}{2}-i(n-i) \right)
$$
and $\tilde{C}(x) = 0.5 - 3x(1-x)$  for $x\in[0,1]$, then

$$
\sup_{x\in [0,1]}\vert \tilde{C}_{n}(\lceil nx \rceil) -\tilde{C}(x) \vert
\le 2n^{-1} \ln n + O(n^{-1}).
$$
\end{lemma}
\begin{proof}
Writing out

$$
\begin{array}{rcl}
\tilde{C}_{n}(i) & = &  n^{-2} \left(i^{2}+i-2iH_{i,1}+(n-i)^{2} + (n-i) -2(n-i)H_{n-i,1} -n^{2}
\right. \\ && \left. 
 - n +2nH_{n,1}+\binom{n}{2} - i(n-i)\right)
\\ & = &
n^{-2}\left(3i^{2}-3in +\frac{1}{2} n^{2} + 2nH_{n,1} - \frac{1}{2}n -2iH_{i,1} -2(n-i)H_{n-i,1}\right)
\\ & < & \frac{1}{2} -3\frac{i}{n}\left(1-\frac{i}{n}\right) +2n^{-1}\ln n 
\end{array}
$$
Therefore, assuming that $1 \le \lceil nx \rceil \le n-1 $

$$
\begin{array}{l}
\vert \tilde{C}_{n}(\lceil nx \rceil) -\tilde{C}(x) \vert \le 3\vert \frac{\lceil nx \rceil}{n}(1-\frac{\lceil nx \rceil}{n})-x(1-x) \vert + 2n^{-1}\ln n
\\ \le \sup\limits_{\vert y-z \vert < 1/n} \vert \tilde{C}(y) -\tilde{C}(z) \vert + 2n^{-1}\ln n \le \frac{6}{n} + 2n^{-1}\ln n + O(n^{-2}).
\end{array}
$$
If $\lceil nx \rceil = n$, we notice that $x \in (1-1/n,1]$ and directly obtain

$$
\begin{array}{l}
\vert \tilde{C}_{n}(\lceil nx \rceil) -\tilde{C}(x) \vert \le 3\vert x(1-x) \vert + 2n^{-1}\ln n \le 2n^{-1}\ln n + \frac{3}{n}.
\end{array}
$$
\end{proof}

\begin{lemma}\label{lemCCdiff}
Define for $i \in \{1,\ldots, n\}$, $T,T' \sim \exp(2)$ 

$$
C_{n}(i,T,T') = \frac{1}{n^{2}} \left(\Expectation{\Phi^{(i)}_{NRE}}+\Expectation{\Phi^{(n-i)}_{NRE}}-\Expectation{\Phi^{(n)}_{NRE}} + \binom{i}{2}T+\binom{n-i}{2}T' \right)
$$
and for $x\in[0,1]$, $T,T' \sim \exp(2)$ 

$$
C(x,T,T') = \frac{1}{2}x^{2}T  + \frac{1}{2}(1-x)^{2}T' - x(1-x)
$$
then

$$
\sup_{x\in [0,1]}\vert C_{n}(\lceil nx \rceil,T,T') -C(x,T,T') \vert
\le 
n^{-1} \ln n + O(n^{-1}) + B_{n},
$$
where $B_{n}$ is a positive random variable that converges to $0$ almost surely with expectation decaying as $O(n^{-1})$
and second moment as $O(n^{-2})$.
\end{lemma}
\begin{proof}
Similarly, as in the proof of Lemma \ref{lemCCDiscdiff} we write out

$$
\begin{array}{rcl}
C_{n}(i,T,T')  & = &
n^{-2}\left( 
\binom{i}{2}T + \binom{n-i}{2}T' 
+\frac{1}{2}(i^{2}+i) - iH_{i,1} + \frac{1}{2}((n-i)^{2}+(n-i)) 
\right. \\&&  \left. 
- (n-i)H_{n-i,1}- \frac{1}{2}(n^{2}-n) + nH_{n,1}\right)
\\& < &
\frac{1}{2}\left(\frac{i}{n}\right)^{2} T + \frac{1}{2}\left(\frac{n-i}{n}\right)^{2} T' -\frac{i}{n}\left(1-\frac{i}{n}\right)
+ n^{-1}\ln n - \frac{1}{2}\left( \frac{i}{n^{2}}T + \frac{n-i}{n^{2}}T' \right).
\end{array}
$$
We denote $A_{n}=\left(1/2\right)\left( \frac{i}{n^{2}}T + \frac{n-i}{n^{2}}T' \right)$ and notice that it converges 
almost surely to $0$ with $n$. Now, assuming that $1 \le \lceil nx \rceil \le n-1 $

$$
\begin{array}{l}
\vert C_{n}(\lceil nx \rceil) - C(x) \vert \le
\frac{1}{2}\vert \left(\frac{\lceil nx \rceil}{n}\right)^{2} - x^{2} \vert T 
\\ + \frac{1}{2}\vert \left(1-\frac{\lceil nx \rceil}{n}\right)^{2} - (1-x)^{2} \vert T' 
+ \vert \frac{\lceil nx \rceil}{n}(1-\frac{\lceil nx \rceil}{n})-x(1-x) \vert + n^{-1}\ln n + A_{n}
\\ <
\sup\limits_{\vert y-z \vert < 1/n}\frac{1}{2}\vert y^{2} - z^{2} \vert T 
+\sup\limits_{\vert y-z \vert < 1/n} \frac{1}{2}\vert y^{2}  - z^{2} \vert T' 
+\sup\limits_{\vert y-z \vert < 1/n} \vert y(1-y)+z(1-z) \vert 
\\ + n^{-1}\ln n + A_{n}
\\ \le
(n^{-1} + O(n^{-2})) T + (n^{-1} + O(n^{-2})) T' + \frac{2}{n} + O(n^{-2})+ n^{-1}\ln n + A_{n}.
\end{array}
$$
If $\lceil nx \rceil = n$, we notice that $x \in (1-1/n,1]$ and directly obtain

$$
\begin{array}{l}
\vert C_{n}(\lceil nx \rceil) - C(x) \vert \le
\frac{1}{2} n^{-2} T  + \frac{1}{2} n^{-2}  T' 
+ n^{-1} + n^{-1}\ln n + A_{n}.
\end{array}
$$
Therefore, if we now denote 

$$
B_{n} = A_{n} + (n^{-1} + O(n^{-2})) T + (n^{-1} + O(n^{-2})) T'
$$
we obtain the statement of the Lemma.
\end{proof}
\end{document}